%% file: lipics-full.tex
\documentclass[a4paper,USenglish,cleveref,autoref,thm-restate]{lipics-v2021}

\hideLIPIcs  

\title{A Faster Algorithm for Fewer Vertex-Disjoint Paths Parameterized by Treewidth}

\author{DongYun Byun}
 {Division of Electrical, Information and Communication Engineering, Kanazawa University, Japan}
{bdy8272@gmail.com}{}
{}

\author{Akira Matsubayashi}
 {Division of Electrical, Information and Communication Engineering, Kanazawa University, Japan
 \and \url{http://carrera.ec.t.kanazawa-u.ac.jp/index_e.html}}
{mbayashi@t.kanazawa-u.ac.jp}{https://orcid.org/0000-0002-7861-4876}
{This work was supported by JSPS KAKENHI Grant Number 23K10984.}

\authorrunning{D.\ Byun and A.\ Matsubayashi}

\Copyright{DongYun Byun and Akira Matsubayashi}

\ccsdesc[500]{Theory of computation~Parameterized complexity and exact algorithms}
\ccsdesc[500]{Mathematics of computing~Paths and connectivity problems}

\keywords{disjoint paths, parameterized algorithm, treewidth, pathwidth,
 strongly exponential time hypothesis}

\nolinenumbers 

\usepackage{framed}

\newtheorem{property}{Property}

\newcommand{\yes}{\texttt{yes}}
\newcommand{\no}{\texttt{no}}
\newcommand{\alg}[1]{\texttt{#1}}
\newcommand{\kDDP}{$k$DDP}
\newcommand{\kUDP}{$k$UDP}
\newcommand{\KKHS}{\textsc{$K\times K$ Hitting Set}}
\newcommand{\KKPHS}{\textsc{$K\times K$ Permutation Hitting Set}}
\newcommand{\aKKPHS}{\textsc{$\alpha K\times K$ Permutation Hitting Set}}
\newcommand{\aKKHS}{\textsc{$\alpha K\times K$ Hitting Set}}
\newcommand{\KKClique}{\textsc{$K\times K$ Clique}}

\newcommand{\TSAT}{$3$-SAT}
\newcommand{\CSAT}{CNF-SAT}
\newcommand{\indeg}[1]{\ensuremath{\mathsf{deg}^-_{#1}}}
\newcommand{\outdeg}[1]{\ensuremath{\mathsf{deg}^+_{#1}}}
\newcommand{\tw}{\mathsf{tw}}
\newcommand{\pw}{\mathsf{pw}}
\newcommand{\rmc}{\ensuremath{\mathsf{rmc}}}
\newcommand{\opt}{\ensuremath{\mathsf{opt}}}
\newcommand{\union}{\ensuremath{\mathsf{union}}}
\newcommand{\ins}{\ensuremath{\mathsf{ins}}}
\newcommand{\glue}{\ensuremath{\mathsf{glue}}}
\newcommand{\gluew}{\ensuremath{\mathsf{glue}_\omega}}
\newcommand{\shift}{\ensuremath{\mathsf{shift}}}
\newcommand{\proj}{\ensuremath{\mathsf{proj}}}
\newcommand{\join}{\ensuremath{\mathsf{join}}}

\newcommand{\nblock}[1]{\ensuremath{\texttt{\#blocks}(#1)}}
\DeclareMathOperator{\cuparrow}{%
\stackrel{\text{\smash{\raisebox{-\height}{$\downarrow$}}}}{\cup}%
}
\DeclareMathOperator*{\bigcuparrow}{%
\stackrel{\text{\smash{\raisebox{-1.3\totalheight}{$\downarrow$}}}}{\bigcup}%
}

\begin{document}
\maketitle

\input{abstract}
\input{section1}
\input{section2}
\input{section3}
\input{section4}
\input{section5}

\bibliography{mainfull}
\end{document}

%% file: abstract.tex
\begin{abstract}
The $k$ vertex-disjoint paths problem asks whether,
 given a graph $G$ and $k$ pairs of vertices
 $(s_1,t_1)$, \ldots, $(s_k,t_k)$,
 $G$ has $k$ pairwise vertex-disjoint paths connecting
 $s_i$ and $t_i$ for all $1\leq i\leq k$.
If $G$ is undirected, then this problem is NP-complete, but
 there exist FPT algorithms parameterized by $k$.
Since these algorithms involve an extremely large function on $k$,
 algorithms for restricted graphs have also been investigated.
In particular,
 a $2^{2\tw\log\tw+\mathcal O(\tw)}\cdot n$ time algorithm
 for undirected graphs with $n$ vertices and treewidth $\tw$
 is proposed by Scheffler (Technical Report 396, TU Berlin, '94), and
 it is proved by
 Lokshtanov, Marx, and Saurabh (SIAM J.\ Comput.\ '18) that, under the ETH,
 there exists no $2^{o(\pw\log\pw)}\cdot n^{\mathcal O(1)}$ time algorithm
 for either directed or undirected graphs with pathwidth $\pw$ and
 for $k=\Omega(\pw^4)$.
It has not been known whether the lower bound also holds for a smaller $k$.
In this paper,
 we prove that, for both the directed and undirected cases,
 there is an algorithm faster than
 Lokshtanov et al.'s lower bound for $k=\tw^{o(1)}$ by proposing
 a $2^{\mathcal O((\tw+k)\log k)}\cdot n$ time algorithm.
Besides,
 we prove a lower bound that, under the SETH,
 there exists no
 $(2-\epsilon)^{\pw\log\pw}\cdot n^{\mathcal O(1)}$ time algorithm
 for directed graphs and for a general $k$.
This lower bound is tight because, with slight modifications,
 Scheffler's algorithm runs in
 $2^{\pw\log\pw+\mathcal O(\pw)}\cdot n$ time also for directed graphs.
\end{abstract}

%% file: section1.tex
\section{Introduction}
The \emph{$k$ (vertex-)disjoint paths problem} asks whether,
 given a graph $G$ and $k$ pairs of vertices
 $(s_1,t_1)$, \ldots, $(s_k,t_k)$, called \emph{demands},
 $G$ has $k$ pairwise vertex-disjoint paths connecting
 $s_i$ and $t_i$ for all $1\leq i\leq k$.
This is a fundamental problem in graph theory, particularly,
 graph homeomorphism and graph minor theory \cite{FHW80,RS95},
 and has applications such as VLSI layout and routing \cite{KLPS90}.
Moreover, this problem can be applied to the recognition of
 network structures admitting Braess's paradox \cite{Ro06,CGS19},
 which is a counterintuitive phenomenon that removal of some links decreases
 the Nash flow cost.

If $G$ is undirected, then this problem can be solved in polynomial time
 for a fixed $k$ \cite{RS95,KKR12}, but
 it is NP-complete for a general $k$ \cite{Karp75}
 and even when $G$ is restricted to a planar graph \cite{Ly75,MP93}.
If $G$ is directed, then this problem is even harder 
 and is NP-complete even for $k=2$, but can be solved in polynomial time
 for DAGs \cite{FHW80}.
Due to the difficulty,
 FPT and XP algorithms for the $k$ disjoint paths problem has extensively been
 investigated.

\subsection{Previous Results}
\subparagraph*{Undirected Case}
The first FPT algorithm parameterized by $k$ for undirected general graphs
 was proposed by Robertson and Seymour \cite{RS95}
 and runs in $f(k)\cdot n^3$ time,
 where $f$ is a certain function and $n$ is the number of vertices.
This running time was later improved by
 Kawarabayashi, Kobayashi, and Reed \cite{KKR12} to
 $g(k)\cdot n^2$ for another function $g$.
According to \cite{AKK+17,Chi23}, however,
 $f$ and $g$ involved in these time complexities are both
 extremely large functions such that $2^{2^{2^{2^{\Omega(k)}}}}$.
This enormous complexity
 provides motivation to design faster algorithms with respect to $k$
 for restricted graphs.
Adler, Kolliopoulos, Krause, Lokshtanov, Saurabh, and Thilikos \cite{AKK+17}
 proposed an algorithm for undirected planar graphs with
 the running time of
 $2^{2^{\mathcal{O}(k^2)}} \cdot n^{\mathcal{O}(1)}$,
 which was later improved by
 Lokshtanov, Misra, Pilipczuk, Saurabh, Zehavi
 \cite{LMPSZ25}
 to
 $2^{\mathcal{O}(k^2)} \cdot n^{\mathcal{O}(1)}$.
As further restriction of graphs, algorithms parameterized by treewidth $\tw$
 and pathwidth $\pw$ have also been considered.
In this direction,
 Scheffler \cite{Scheffler94}
 proposed a $2^{\mathcal{O}(\tw\log\tw)} \cdot n$ time algorithm
 that takes $G$ and its tree decomposition of width $\tw$ as input.
This algorithm achieves the running time independent of $k$,
 which is based on the fact
 that at most $\tw+1$ disjoint paths can pass through one bag
 in a tree decomposition.
By analyzing Scheffler's algorithm in detail,
 we can estimate its specific time complexity as
 $(3\tw + 3)^{2\tw + 2}\cdot\mathcal O(\tw\cdot n)
=2^{2\tw\log\tw+\mathcal O(\tw)}\cdot n$.\footnote{%
In this paper, the base of the logarithm is assumed to be $2$.
}
Lokshtanov, Marx, and Saurabh
 \cite{LMS18}
 proved that if the exponential time hypothesis (ETH) is true,
 then no $2^{o(\pw\log\pw)}\cdot n^{\mathcal{O}(1)}$ time algorithm exists.
This lower bound implies that
 Scheffler's algorithm achieves a tight time complexity under the ETH.
We note that the proof of the lower bound requires $k=\Omega(\pw^4)$.
It has not been known whether the lower bound holds for a smaller $k$.

\subparagraph*{Directed Case}
Since
 the $k$ disjoint paths problem is NP-complete
 even for $k=2$ \cite{FHW80},
 unless $\text P=\text{NP}$, no FPT algorithm parameterized by $k$ exists.
For directed planar graphs,
 an $n^{\mathcal{O}(k)}$ time XP algorithm was proposed by
 Schrijver \cite{Schrijver94}, and later
 Cygan et al.\ improved this algorithm to
 an FPT algorithm with the running time of
 $2^{2^{\mathcal{O}(k^2)}} \cdot n^{\mathcal{O}(1)}$ \cite{CMPP13}.
Lokshtanov et al.'s lower bound \cite{LMS18} was also proved
 for directed graphs and for $k=\Omega(\pw^4)$.
As a lower bound for a small $k$,
 we can show that, under the ETH,
 there exists no $2^{o(\pw)}\cdot n^{\mathcal O(1)}$ time algorithm for
 any $k\geq 2$.
This is because,
 in the proof of the NP-completeness for $k=2$ in \cite{FHW80},
 the graph reduced from a $3$-CNF formula input to \TSAT{}
 has the pathwidth $\pw=\mathcal O(m)$, where
 $m$ is the number of clauses.
By the Sparsification Lemma \cite{IPZ01},
 for any $\varepsilon>0$,
 any $3$-CNF formula $f$ with $n$ variables and $m$ clauses
 can be transformed into a disjunction of $2^{\varepsilon n}$ $3$-CNF formulas
 $f_1$, \ldots, $f_{2^{\varepsilon n}}$
 each of which has $n$ variables and $c(\varepsilon)\cdot n$ clauses
 in $2^{\varepsilon n}\cdot n^{\mathcal O(1)}$ time,
 where $c(\varepsilon)$ is a certain function of $\epsilon$.
If there is a $2^{o(\pw)}$ time algorithm for the $2$ disjoint paths problem,
 then through the reduction of \cite{FHW80},
 we can decide
 the satisfiability of $f$ by deciding
 the satisfiabilities of all $f_1$, \ldots, $f_{2^{\varepsilon n}}$ in
 $2^{\varepsilon n}\cdot n^{\mathcal O(1)}
+2^{\varepsilon n}\cdot 2^{o(c(\varepsilon)\cdot n)}=2^{\epsilon n+o(n)}$ time
 for an arbitrarily small $\varepsilon>0$, which violates the ETH.

\subsection{Our Contribution}
In this paper, we present two results for the $k$ disjoint paths problem
 on directed graphs with the parameter of treewidth.
The problem is formulated as follows.

\begin{oframed}
\noindent\textsc{$k$ Directed (Vertex-)Disjoint Paths} (\kDDP)
 Parameterized by Treewidth
\begin{description}
\item[Input]
An $n$-vertex directed graph $G$,
 demands $(s_1,t_1)$, \ldots, $(s_k,t_k)$, and
 a tree decomposition of $G$ with width $\tw$.
Here,
 the vertices $s_1$, \ldots, $s_k$, $t_1$, \ldots, $t_k$ are all distinct.
\item[Question]
Does $G$ have $k$ pairwise vertex-disjoint paths connecting
 from $s_i$ to $t_i$ for all $1\leq i\leq k$?
\item[Parameter] $\tw$.
\end{description}
\end{oframed}

\subsubsection{Fast Algorithm for Small $k$}
First, we show that there is an algorithm faster than the lower bound
 of \cite{LMS18} for a small $k=\tw^{o(1)}$ by proving
 Theorem~\ref{th:Algorithm_tw} below.
If $k=\mathcal O(1)$, in particular,
 then this algorithm is a single exponential FPT algorithm
 with respect to $\tw$.
\begin{theorem}
\label{th:Algorithm_tw}
\kDDP{} can be solved in
 $2^{\omega\cdot (\tw + 2k)\log k+\mathcal{O}(\tw + k)} \cdot n$ time, where
 $\omega$ is the matrix multiplication exponent.\footnote{%
The best upper bound currently known is $2.371339$ \cite{ADVXXZ2025}.}
\end{theorem}
To prove Theorem~\ref{th:Algorithm_tw},
 we design an algorithm
 using
 the \emph{rank based approach} proposed by Bodlaender et al.\ \cite{BCKN15}.
This approach is based on dynamic programming that proceeds while maintaining
 a set $\mathcal A$ of partitions of the vertex set contained in a bag of
 the tree decomposition.
Each partition in $\mathcal A$ corresponds to a partial solution in each step
 of the dynamic programming.
Because the maximum size of such a partition set is
 $\tw^{\Theta(\tw)}=2^{\Theta(\tw\log\tw)}$ in general,
 a naive approach would need $2^{\Theta(\tw\log\tw)}$ time
 to update the partition set for each step of the dynamic programming.
In \cite{BCKN15},
 $2^{\mathcal O(\tw)}\cdot n$ time algorithms are designed
 for \emph{single connectivity problems}, i.e.,
 problems for finding a single connected component satisfying
 certain conditions,
 such as
 the traveling salesman problem (TSP) and the Steiner tree problem.
The single exponential time is achieved through
 maintaining a partition set $\mathcal A'\subseteq\mathcal A$ only of the size
 $2^{\tw}$.
The correctness of the reduction of $\mathcal A$ is based on two facts:
One is that for each partition $q$,
 retaining information of partitions $p\in\mathcal A$ such that
 their join $p\sqcup q$
 is a partition consisting of a single block
 (i.e., a single connected component) is enough.
The other is that a matrix representing such information can be factorized into
 two matrices of a rank at most $2^{\tw}$.
More strictly,
 $\mathcal A'$ needs to satisfy a property that
 for a partition $q$,
 if there is $p\in\mathcal A\setminus\mathcal A'$ such that $p\sqcup q$
 is a single block partition, then
 there is $p'\in\mathcal A'$ such that $p'\sqcup q$
 is a single block partition
 (which is said that \emph{$\mathcal A'$ represents $\mathcal A$}).
Hence, in each step of the dynamic programming,
 we need to update $\mathcal A'$ using only operations preserving this property.
Such operations are said to \emph{preserve representation}.

In our proof of Theorem~\ref{th:Algorithm_tw},
 we utilize the fact that the rank based approach can also be applied to
 \emph{multi-connectivity problems}\footnote{%
In this paper, multi-connectivity problems refer to problems of finding
 multiple connected components
 each of which is required to contain a prescribed vertex,
 such as \kDDP.
}
 through a natural extension.
%
%
We add the $k$ demands $(s_1,t_1)$, \ldots, $(s_k,t_k)$ to all bags
 to maintain appropriate partial solutions of \kDDP, and
 only consider partitions consisting of $k$ or more blocks,
 $k$ of which each contain $s_1$, \ldots, $s_k$.
Let $\Pi_k$ be the set of all such partitions.
The notion of representation is generalized to sets of only partitions in $\Pi_k$.
More specifically, we define that a set $\mathcal A$ of partitions in $\Pi_k$
 is represented by $\mathcal A'\subseteq\mathcal A$ as follows:
For any partition $q\in\Pi_k$, if there is $p\in\mathcal A\setminus\mathcal A'$
 such that their join $p\sqcup q$ is a partition
 consisting of exactly $k$ blocks, each containing $s_i$, then
 there is $p'\in\mathcal A'$ such that
 $p'\sqcup q$ is a partition
 consisting of exactly $k$ blocks, each containing $s_i$.
This condition
 can be viewed as
 a relaxation of the condition that
 $p\sqcup q$ and $p'\sqcup q$ correspond to solutions of \kDDP,
 because, actually, a partition consisting of $k$ blocks,
 each containing both $s_i$ and $t_i$, corresponds to a solution of \kDDP.
Although addition of the $k$ demands to all bags increases the maximum size of
 $\mathcal A$ to $(\tw+2k)^{\Theta(\tw+k)}$,
 we can show
 by only considering partitions in $\Pi_k$
 that a matrix of all partitions $p\in\mathcal A$ and $q\in\Pi_k$
 can be factorized into two matrices of a rank at most 
 $k^{\mathcal O(\tw+k)}=2^{\mathcal O(\tw+k)\log k}$, by which
 $\mathcal A'$ of the same size can be computed,
 and that
 operations to update $\mathcal A'$ in each step of the dynamic programming
 can be generalized to preserve the property of the generalized representation.
The detailed implementation and analysis of the proposed algorithm is
 mainly inspired by the algorithm for TSP in \cite{BCKN15}.

We note that with slight modifications of our proof,
 Theorem~\ref{th:Algorithm_tw} also holds
 for the undirected version of \kDDP,
 the \textsc{$k$ Undirected Disjoint Paths} (\kUDP).
Moreover,
 Theorem~\ref{th:Algorithm_tw} holds for the edge-weighted versions
 of both \kDDP{} and \kUDP,
 which require
 the minimum sum of edge-weights of disjoint paths
 in a given edge-weighted graph,
 as well as
 Scheffler's algorithm
 being extensible to the edge-weighted \kUDP{} \cite{Scheffler94}.
In Section~\ref{sc:Algorithm_tw},
 we design an algorithm for the edge-weighted \kDDP.

\subsubsection{Tight Lower Bound under the SETH for General $k$}
Second, we prove a precise lower bound for \kDDP{}
 under the Strongly Exponential Time Hypothesis (SETH) as follows.
\begin{theorem}
\label{th:LowerBound_SETH}
If the SETH is true, then for any $\varepsilon>0$,
 there exists no $(2 - \varepsilon)^{\pw\log\pw}\cdot n^{\mathcal O(1)}$ time
 algorithm for \kDDP{} with a general $k$.
\end{theorem}
The idea of our proof of Theorem~\ref{th:LowerBound_SETH} is based on
 the lower bound of \cite{LMS18} under the ETH.
In \cite{LMS18}, \TSAT{} is reduced to \kDDP{} via
 the following problem.\footnote{%
In \cite{LMS18}, a reduction from
 a non-permutation version of \KKHS, where
 $S$ needs to contain exactly one element from each row but not necessarily from
 each column, to \kDDP{} was presented.
In fact,
 the permutation condition
 is essential for the reduction.
Nevertheless,
 replacing with \KKPHS{} in the reduction,
 the lower bound for \kDDP{} in \cite{LMS18} is correctly proved,
 since a lower bound for \KKHS{} was proved
 for the permutation version in \cite{LMS18}.
%
}
\begin{oframed}
\noindent\KKPHS
\begin{description}
\item[Input]
Nonempty sets $S_1$, \ldots, $S_m\subseteq [K]\times [K]$, where
 $[K]$ is the set of integers from $1$ to $K$.
\item[Question]
Is there a set $S$ containing exactly one element from
 each row and each column of the grid on $[K]\times [K]$ such that
 $S\cap S_i\neq\emptyset$ for each $1\leq i\leq m$?
\end{description}
\end{oframed}
In the proof of \cite{LMS18}, first,
 \TSAT{} is reduced to \KKPHS{}
 such that
 $K\log K=\mathcal O(n)$ for the number $n$ of variables in a $3$-CNF formula
 input to \TSAT.
The most complicated part in the reduction is
 \emph{permutationization} of $S$, i.e.,
 restricting the condition of $S$ to choosing
 exactly one element from each column as well as each row.
Then,
 \KKPHS{} is reduced to \kDDP{} through constructing
 an input graph $G$ of \kDDP{} with the pathwidth $\pw=\mathcal O(K)$.
Consequently, it follows that
 $\pw\log\pw=\mathcal O(K\log K)=\mathcal O(n)$, which implies
 the lower bound of \cite{LMS18}.

To prove Theorem~\ref{th:LowerBound_SETH},
 we need to transform a general CNF formula with $n$ variables
 to an input graph $G$ of \kDDP{} with $\pw\log\pw=n+o(n)$.
For this purpose, first,
 we choose $K$ and a positive integer $\alpha$ such that $\alpha K\log K=n+o(n)$,
 and transform the CNF formula to an input of a problem, called \aKKPHS,
 in which $[K]\times [K]$ is replaced with $[\alpha K]\times [K]$, and
 the condition of $S$ is changed to choosing exactly one element from
 each row and each column in each of $\alpha$ grids on $[K]\times [K]$
 obtained by dividing $[\alpha K]\times [K]$.
We use the permutationization technique of \cite{LMS18} for this reduction.
Then, we transform the input of \aKKPHS{} to an input graph $G$ of \kDDP{}
 with $\pw=\alpha K+\mathcal O(K)$, using a slightly modified
 technique of the transformation of \cite{LMS18}.
By choosing $\alpha=n^{o(1)}$ and $K\approx\frac{n}{\alpha\log n}$,
 $\pw\log\pw=n+o(n)$ can be satisfied,
 which implies Theorem~\ref{th:LowerBound_SETH}.

The lower bound of Theorem~\ref{th:LowerBound_SETH} is tight.
Actually, Scheffler's algorithm \cite{Scheffler94}, designed for \kUDP{},
 runs in $2^{\pw\log\pw+\mathcal O(\pw)}\cdot n$ time for \kDDP{}
 with slight modifications.
We describe the proof and the tightness of Theorem~\ref{th:LowerBound_SETH}
 in Section~\ref{sc:LowerBound_SETH}.

%% file: section2.tex
\section{Preliminaries}
\subparagraph*{Graphs}
Unless otherwise stated, we assume graphs are directed.
For a graph $G$ and a subset $X$ of the edge set $E(G)$ of $G$,
 the subgraph of $G$ induced by $X$ is denoted by $G[X]$.
The indegree and outdegree of a vertex $v\in V(G)$ are denoted by
 $\indeg{G}(v)$ and $\outdeg{G}(v)$, respectively.
We call a (simple) path from a vertex $u$ to a vertex $v$ a \emph{$uv$-path}.
A \emph{linear forest} is the disjoint union of some vertex-disjoint paths.

\subparagraph*{Tree Decomposition and Treewidth}
A \emph{tree decomposition} $\mathcal{T}$ of a graph $G$ \cite{RS84}
 is a pair $(T,\{B_t\}_{t\in V(T)})$ of an undirected tree $T$ and
 a family of a \emph{bag},
 a subset $B_t\subseteq V(G)$ of vertices in $G$,
 associated with each node $t$ in $T$, such that the following
 conditions are satisfied.
\begin{itemize}
\item
$\bigcup_{t\in V(T)} B_t = V(G)$.
\item
For every edge $e\in E(G)$, there is a bag $B_t$
 containing both the end-vertices of $e$.
\item
If $v\in B_x$ and $v\in B_y$, then
 for every node $z$ in the unique $xy$-path in $T$, $v\in B_z$.
\end{itemize}
The \emph{width of a tree decomposition $\mathcal{T}$}
 is $\max_{t\in V(T)} |B_t| - 1$, and
 the \emph{treewidth $\tw(G)$ of a graph $G$}, or simply $\tw$,
 is the minimum width of over all tree decompositions of $G$.
While computing a tree decomposition of width $\tw(G)$
 is generally NP-hard \cite{ACP87},
 it is known that for a fixed $k$,
 a tree decomposition of width $k$ (if exists) 
 can be computed in polynomial time \cite{Bod96}.
If $T$ is a path, then $\mathcal T$ is called a \emph{path decomposition}.
The \emph{pathwidth $\pw(G)$ of a graph $G$}, or simply $\pw$,
 is the minimum width of over all path decompositions of $G$.

In dynamic programming based on a tree decomposition,
 a tree decomposition satisfying special conditions,
 which is said to be \emph{nice}, is often used.
In this paper, we use nice tree decompositions defined in \cite{BCKN15}.
\begin{definition}[Nice tree decomposition]
\label{df:NiceTreeDecomposition}
A tree decomposition $(T, \{B_t\}_{t\in V(T)})$ is said to be \emph{nice}
 if $T$ is a rooted tree, and every node $x$ is one of the following types.
\begin{description}
\item[Leaf Node]
The node $x$ is a leaf of $T$ and satisfies $B_x = \emptyset$.
\item[Introduce Vertex Node]
The node $x$ has a unique child $y$
 and satisfies $B_x = B_y \cup \{v\}$ for some vertex $v \in B_x$.
The node $x$ is said to \emph{introduce the vertex $v$}.
\item[Introduce Edge Node]
The node $x$ has a unique child $y$
 and satisfies $u,v\in B_x$ and $B_x = B_y$ for some edge $uv\in E(G)$.
The node $x$ is said to \emph{introduce the edge $uv$}.
\item[Forget Node]
The node $x$ has a unique child $y$
 and satisfies $B_x = B_y \setminus \{v\}$ for some vertex $v \in B_y$.
The node $x$ is said to \emph{forget the vertex $v$}.
\item[Join Node]
The node $x$ has two children $y, z$ and satisfies $B_x = B_y = B_z$.
\end{description}
In addition,
 every edge of $G$ must be introduced exactly once.
\end{definition}
\begin{proposition}[\cite{BCKN15}]
\label{pr:NiceTreeDecomposition}
For a given $n$-vertex graph $G$ and its tree decomposition of width $\tw$,
 a nice tree decomposition of width $\tw$ rooted
 by a forget node with an empty bag or
 by an arbitrary introduce edge node
 can be computed in $\tw^{\mathcal{O}(1)} \cdot n$ time.
Moreover,
 the computed nice tree decomposition has $O(n)$ nodes.
\end{proposition}
For a nice tree decomposition $\mathcal{T}$ of $G$ and its node $x$,
 let $V_x \subseteq V(G)$ be the set of all vertices contained in
 $x$ or a descendant node of $x$.
Besides, let $E_x\subseteq E(G)$ be the set of all edges
 introduced by $x$ or a descendant node of $x$, and
 let $G_x = (V_x, E_x)$.

\subparagraph*{Partition}
A \emph{partition} $p$ of a set $U$ is a family of subsets of $U$
 satisfying the following conditions.
\begin{itemize}
\item
The family $p$ has no empty set.
\item
The union of all elements of $p$ equals $U$.
\item
All elements of $p$ are pairwise disjoint.
\end{itemize}
An element of a partition is called a \emph{block}.
The set of all partitions of a set $U$ is denoted by $\Pi(U)$.
The partition set $\Pi(U)$,
 together with the binary relation $\preceq$ representing
 the \emph{refinement} of a partition, forms a lattice $(\Pi(U),\preceq)$.
Here,
 $p\preceq q$ represents that every block in $p$ is a subset
 of some block in $q$. 
The lattice $(\Pi(U),\preceq)$ has
 the trivial partition $\{U\}$ as the greatest element and
 the partition consisting of singleton sets
 $\bigcup_{u\in U}\{\{u\}\}$ as the least element.
In the lattice $(\Pi(U),\preceq)$,
 the least upper bound and the greatest lower bound
 of two partitions $p$ and $q$
 are called
 the \emph{join of $p$ and $q$}, denoted by $p\sqcup q$, and
 the \emph{meet of $p$ and $q$}, denoted by $p\sqcap q$, respectively.
In this paper, following the notation of \cite{BCKN15},
 we use $q\sqsubseteq p$ instead of $p\preceq q$.

%% file: section3.tex
\section{Fast Algorithm for Small $k$}
\label{sc:Algorithm_tw}
In this section, we design and analyze our algorithm to prove
 Theorem~\ref{th:Algorithm_tw}.
This algorithm
 performs dynamic programming similar to that of \cite{BCKN15} for the TSP.
In particular,
 we represent and retain
 linear forests that are partial solutions of \kDDP{} for each bag
 using partitions with blocks that consist of
 end-vertices of the linear forests contained in the bag.
Since, in our algorithm,
 all the demand vertices
 $s_1$, \ldots, $s_k$, $t_1$, \ldots, $t_k$
 are added to every bag,
 any linear forest composing a solution is retained without
 being forgotten.
In each step of dynamic programming,
 the retained set of partitions is updated and then
 reduced to a size of at most $(k+1)^{\tw+k+1}$
 using the rank based approach.

In Section~\ref{ssc:Algorithm_tw_Notation},
 we describe the notation used in Section~\ref{sc:Algorithm_tw},
 all of which were used in \cite{BCKN15}.
In Section~\ref{ssc:k-RankBasedApproach},
 we generalize the rank based approach
 to multi-connectivity problems in a natural way.
Using the generalized approach,
 we design and analyze our algorithm for solving \kDDP{}
 in Section~\ref{ssc:k-Algorithm_tw}.

\subsection{Notation}
\label{ssc:Algorithm_tw_Notation}
For a partition $p$,
 the number of blocks $|p|$ is more specifically denoted by $\nblock{p}$.
Suppose that $U$ is a set and $p\in\Pi(U)$ is a partition of $U$.
For a subset $X$ of $U$,
 let $p_{\downarrow X} \in \Pi(X)$ be the partition obtained by
 removing every element in $U\setminus X$ from each block in $p$, and
 let $U[X]\in\Pi(U)$ be the partition consisting of the block $X$
 and the singleton $\{e\}$ for every element $e$ in $U\setminus X$.
For a superset $Y$ of $U$,
 let $p_{\uparrow Y} \in \Pi(Y)$ to the partition obtained by
 adding every element $e$ in $Y\setminus U$ to $p$ as the singleton block
 $\{e\}$.

For sets $X$ and $Y$,
 $X^Y$ is the set of all mappings from $Y$ to $X$.
For a mapping $f$, the mapping obtained by restricting
 the domain to $X$ is denoted by $f|_X$.
We define
 $f[v\rightarrow \alpha]$ as the mapping
 obtained by replacing $f(v)$ with $\alpha$, i.e.,
 $f\setminus \{ (v, f(v)) \} \cup \{ (v, \alpha) \}$.
(Here,
 we identify $f(x)=y$ with the pair $(x,y)$.
If $(v, f(v))$ does not exists, then we simply add $(v, \alpha)$.)
For two functions $f$ and $g$ with the same domain,
 we define $(f+g)(x) = f(x) + g(x)$.
For the operator $\min$ that returns the minimum value of a set of numbers,
 we define $\min(\emptyset) = \infty$.

\subsection{Extension of Rank Based Approach to Multi-Connectivity Problems}
\label{ssc:k-RankBasedApproach}
In Section~\ref{sssc:Operators},
 we introduce operators for weighted partition sets, which are used to
 describe processes in each step of our dynamic programming.
The operators were all defined in \cite{BCKN15}.
In Section~\ref{sssc:k-Reduce},
 we prove that
 the operators preserve representation also for $k$-block partition sets,
 and that for any set $\mathcal A$
 of partitions consisting of $k$ or more blocks,
 a subset $\mathcal A'\subseteq\mathcal A$ of a size at most $(k+1)^{\tw+k+1}$
 can be obtained.
The algorithm for finding $\mathcal A'$ is used to reduce
 a partition set produced by one of the operators.

\subsubsection{Weighted Partition Sets and Operators}
\label{sssc:Operators}
\begin{definition}[Weighted partition sets]
For a set $U$,
 a \emph{weighted partition set}
 $\mathcal{A} \subseteq \Pi(U) \times \mathbb{N}$
 is a set of pairs of a partition of $U$ and a nonnegative integer.
If $w=0$ for every $(p, w) \in \mathcal{A}$, then
 $\mathcal{A}$ is \emph{unweighted}.
\end{definition}
\begin{definition}[Operators on weighted partition sets]
\label{df:Operators}
Suppose that $U$ is a set and
 $\mathcal{A}\subseteq \Pi(U) \times \mathbb{N}$.
\begin{description}
\item[rmc]
$\rmc(\mathcal{A})
 =\left\{(p,w)\in\mathcal{A}\mid\forall w'<w:(p,w')\notin\mathcal{A}\right\}$,
 i.e.,
$\rmc(\mathcal{A})$ removes every partition $p$
 with a non-minimum weight from $\mathcal{A}$.
\item[Union]
For $\mathcal{B}\subseteq\Pi(U)\times\mathbb{N}$,
 we define
 $\mathcal{A}\cuparrow\mathcal{B}=\rmc(\mathcal{A}\cup\mathcal{B})$, i.e.,
 $\mathcal{A}\cuparrow\mathcal{B}$
 extracts only partitions with minimum weights from the union
 of two weighted partition sets $\mathcal A$ and $\mathcal B$.
\item[Insert]
For a set $X$ disjoint with $U$,
 we define
 $\ins(X,\mathcal{A})=\{(p_{\uparrow U\cup X},w)\mid (p,w)\in\mathcal{A}\}$, i.e,
 $\ins(X,\mathcal{A})$ adds every element $e$ in $X$
 to each weighted partition in $\mathcal A$
 as the singleton block $\{e\}$.
\item[Shift]
For $w' \in \mathbb{N}$,
 we define
 $\shift(w',\mathcal{A})=\{(p,w+w')\mid (p,w)\in\mathcal{A}\}$, i.e.,
 $\shift(w',\mathcal{A})$ increases the weight of every partition
 in $\mathcal A$ by $w'$.
\item[Glue]
For elements $u$ and $v$,
 let $\hat{U} = U\cup \{ u, v \}$.
We define
 $\glue(uv,\mathcal{A})\subseteq\Pi(\hat{U})\times\mathbb{N}$ as
\[
\glue(uv, \mathcal{A})
 =
\rmc\left( \left\{ 
 \left(\hat{U}[\{u,v\}]\sqcup p_{\uparrow\hat{U}}, w\right)\mid
 (p,w)\in\mathcal{A}
\right\} \right),
\]
i.e.,
 $\glue(uv,\mathcal{A})$ adds $e\in\{u,v\}\setminus U$
 to each weighted partition in $\mathcal A$ as the singleton block $\{e\}$,
 replaces two blocks in the partition containing $u$ and $v$ with their union
 (if necessary),
 and then
 extracts partitions with minimum weights.
In addition,
 for $\omega:\hat{U}\times\hat{U}\rightarrow\mathbb{N}$,
 we define
 $\gluew(uv,\mathcal{A})=\shift(\omega(u,v),\glue(uv,\mathcal{A}))$.
\item[Project]
For a subset $X$ of $U$, let $\bar{X}=U\setminus X$.
We define
 $\proj(X, \mathcal{A}) \subseteq \Pi(\bar{X}) \times \mathbb{N}$ as
%
\[
\proj(X,\mathcal{A})=\rmc\left(\left\{
 (p_{\downarrow \bar{X}},w)\mid (p, w)\in\mathcal{A}
 \wedge \forall e\in X: \exists e'\in \bar{X}: p\sqsubseteq U[\{e,e'\}]
\right\}\right),
\]
i.e.,
 $\proj(X, \mathcal{A})$ removes elements in $X$ from each weighted partition
 in $\mathcal A$, and then
 extracts partitions with minimum weights, such that
 the number of blocks was not decremented by the removal of elements in $X$.
\item[Join]
For a set $U'$, let $\hat{U} = U\cup U'$.
For a weighted partition set
 $\mathcal{B} \subseteq \Pi(U') \times \mathbb{N}$,
 we define
 $\join(\mathcal{A},\mathcal{B})\subseteq\Pi(\hat{U})\times\mathbb{N}$ as
\[
\join(\mathcal{A},\mathcal{B})=\rmc\left(\left\{
 (p_{\uparrow\hat{U}}\sqcup q_{\uparrow\hat{U}},w_1+w_2)\\
 \mid (p,w_1)\in\mathcal{A},(q,w_2)\in\mathcal{B}
 \right\}\right),
\]
i.e.,
 $\join(\mathcal{A},\mathcal{B})$ adds
 elements in $U'\setminus U$ to each partition $p$ in $\mathcal A$ and
 elements in $U\setminus U'$ to each partition $q$ in $\mathcal B$,
 joins $p$ and $q$ while adding their weights,
 and then
 extracts partitions with minimum weights.
\end{description}
\end{definition}
\begin{proposition}[\cite{BCKN15}]
\label{pr:OperatorComplexity}
Operators $\union$, $\ins$, $\shift$, $\glue$, and $\proj$
 can be performed in $I \cdot |U|^{\mathcal{O}(1)}$ time, where
 $I$ is the size of an input for each operator.
The operator $\join(\mathcal{A},\mathcal{B})$ can be performed
 in $|\mathcal{A}|\cdot|\mathcal{B}|\cdot|U|^{\mathcal{O}(1)}$ time.
\end{proposition}

\subsubsection{Reducing Partition Sets for Multi-Connectivity Problems}
\label{sssc:k-Reduce}
We suppose that a set $U$ contains prescribed elements $s_1$, \ldots, $s_k$
 (which represent origins of $k$ demands for \kDDP).
Let $\Gamma_k(U)$ be the set of all partitions of $U$ consisting of
 $k$ blocks $b_1$, \ldots, $b_k$ such that $s_i\in b_i$ for each 
 $1\leq i\leq k$.
If $k=1$, then $\Gamma_1(U)=\{U\}$.
We define $\Pi_k(U)=\{p\mid\exists q\in\Gamma_k(U): q\sqsubseteq p\}$, i.e.,
 $\Pi_k(U)$ is the set of partitions that are
 refinements of partitions in $\Gamma_k(U)$.
\begin{definition}[$k$-Representation]
\label{def:representation}
For a weighted partition set $\mathcal{A} \subseteq \Pi_k(U)\times\mathbb{N}$
 and a partition $q\in \Pi_k(U)$, we define
\[
\opt_k(q, \mathcal{A})=\min
 \left\{w\mid\exists (p,w)\in\mathcal{A}: p\sqcup q\in\Gamma_k(U)\right\}.
\]
If a subset $\mathcal{A}'$ of $\mathcal A$ satisfies
 $\opt_k(q, \mathcal{A}') = \opt_k(q, \mathcal{A})$ for every $q\in \Pi_k(U)$,
 then
 we define that $\mathcal{A}'$ \emph{$k$-represents} $\mathcal{A}$.
\end{definition}
\begin{definition}[Preserving $k$-representation]
\label{def:preserving representation}
Let
$U'\supseteq\{s_1,\ldots,s_k\}$ and $Z$ be sets, and
 let $f \colon 2^{\Pi_k(U)\times\mathbb{N}} \times Z
 \rightarrow 2^{\Pi_k(U')\times\mathbb{N}}$ be a function.
If,
 for every combination of
 $\mathcal{A}\subseteq \Pi_k(U)\times\mathbb{N}$,
 $\mathcal{A}'\subseteq\mathcal A$
 that $k$-represents $\mathcal{A}$, and $z\in Z$,
 $f(\mathcal{A}',z)$ also $k$-represents $f(\mathcal{A}, z)$,
 then we define that
 \emph{$f$ preserves $k$-representation}.
\end{definition}

\begin{lemma}
\label{lm:Representation}
Operators
 $\union$, $\ins$, $\shift$, $\glue$, $\proj$, $\join$
 in Definition~\ref{df:Operators}
 preserves $k$-representation.
\end{lemma}
\begin{proof}
The lemma is proved similarly to the lemma for $k=1$ in \cite{BCKN15}.
Suppose that for a weighted partition set
 $\mathcal{A}\subseteq\Pi_k(U)\times\mathbb N$,
 a subset $\mathcal{A}'$ of $\mathcal A$ $k$-represents $\mathcal{A}$.
\begin{description}
\item[Union ($\mathcal A\cuparrow\mathcal B$)]
For any $\mathcal{B}\subseteq\Pi_k(U)\times\mathbb{N}$ and $q\in\Pi_k(U)$,
 it follows that
\[
\begin{split}
\opt_k(q,\mathcal{A}'\cuparrow\mathcal{B})
& =\min\left\{\opt_k(q,\mathcal{A}'),\opt_k(q,\mathcal{B})\right\}\\
& =\min\left\{\opt_k(q,\mathcal{A}),\opt_k(q,\mathcal{B})\right\}
 =\opt_k(q,\mathcal{A}\cuparrow\mathcal{B}).
\end{split}
\]
Therefore, $\union$ preserves $k$-representation.
\item[Insert ($\ins(X,\mathcal A)$)]
We may assume that $X$ consists of a single element $e\notin U$.
For otherwise,
 there exists a partition $\{Y,Z\}$ of $X$, and it follows that
$\ins(Y\cup Z, \mathcal{A}) = \ins(Y, \ins(Z, \mathcal{A}))$.
Let
 $p\in \Pi_k(U)$ and
 $q\in \Pi_k(U\cup \{e\})$.
\begin{enumerate}
\item
If $\{e\} \in q$, then
 $\{e\} \in p_{\uparrow U\cup\{e\}} \sqcup q$.
Therefore,
 $p_{\uparrow U\cup\{e\}} \sqcup q$ has
 either at least $k+1$ blocks
 or a block containing both
 $s_i$ and $s_j$ for some $i\neq j$.
In either case,
 $p_{\uparrow U\cup\{e\}} \sqcup q$ is not contained in $\Gamma_k(U)$.
This means that
 $\opt_k(q, \ins(\{e\}, \mathcal{A})) =
 \opt_k(q, \ins(\{e\}, \mathcal{A}')) = \infty$.
\item
If $\{e\} \notin q$, then
 $p_{\uparrow U\cup\{e\}}\sqcup q$ does not contain
 $e$ as the singleton block $\{e\}$.
Therefore,
 it follows that
 $\nblock{p_{\uparrow U\cup\{e\}} \sqcup q}= \nblock{p\cup q_{\downarrow U}}$,
 and that for any $\mathcal{C}\subseteq \Pi_k(U)\times\mathbb{N}$,
\[
\begin{split}
\opt_k(q,\ins(\{e\},\mathcal{C}))
& =\min\left\{w\mid\exists (p,w)\in\mathcal{C}:
 p_{\uparrow U\cup\{e\}}\sqcup q\in\Gamma_k(U\cup\{e\})\right\}\\
& =\min\left\{w\mid\exists (p,w)\in\mathcal{C}:
 p\sqcup q_{\downarrow U}\in \Gamma_k(U)\right\}
 =\opt_k(q_{\downarrow U}, \mathcal{C}).
\end{split}
\]
Hence, $\ins$ preserves $k$-representation as
\[
\opt_k(q,\ins(\{e\},\mathcal{A}'))
 =\opt_k(q_{\downarrow U},\mathcal{A}')
 =\opt_k(q_{\downarrow U},\mathcal{A})
 =\opt_k(q,\ins(\{e\},\mathcal{A})).
\]
\end{enumerate}
\item[Shift ($\shift(w',\mathcal A)$)]
For any $q \in \Pi_k(U)$ and $w' \in \mathbb{N}$, it follows that
\[
\opt_k(q,\shift(w',\mathcal{A}'))
 =w'+\opt_k(q,\mathcal{A}')
 =w'+\opt_k(q,\mathcal{A})
 =\opt_k(q,\shift(w',\mathcal{A})).
\]
Therefore, $\shift$ preserves $k$-representation.
\item[Glue ($\glue(uv,\mathcal A)$, $\gluew(uv,\mathcal A)$)]
We may assume $u, v \in U$.
For otherwise, it follows that
 $\glue(uv, \mathcal{A}) = \glue(uv, \ins(\{u, v\}\setminus U, \mathcal{A}))$.
Here, we note that $\ins$ preserves $k$-representation.
For any $\mathcal{C}\subseteq \Pi_k(U)\times\mathbb{N}$ and
 $q\in \Pi_k(U)$, it follows that
\[
\begin{split}
\opt_k(q,\glue(uv, \mathcal{C}))
& =\min\left\{w\mid\exists (p,w)\in\mathcal{C}: (p\sqcup U[\{u,v\}]) 
 \sqcup q\in\Gamma_k(U)\right\}\\
& =\min\left\{w\mid\exists (p,w)\in\mathcal{C}: p\sqcup 
 (q\sqcup U[\{u,v\}])\in\Gamma_k(U)\right\}\\
& =\opt_k(q\sqcup U[\{u,v\}],\mathcal{C}).
\end{split}
\]
Hence, $\glue$ preserves $k$-representation as
\[
\begin{split}
\opt_k(q, \glue(uv, \mathcal{A}'))
& =\opt_k(q\sqcup U[\{u,v\}],\mathcal{A}')\\
& =\opt_k(q\sqcup U[\{u,v\}],\mathcal{A})
 =\opt_k(q,\glue(uv,\mathcal{A})).
\end{split}
\]
Since both $\shift$ and $\glue$ preserve $k$-representation,
 $\gluew$ also preserves $k$-representation.
\item[Project ($\proj(X,\mathcal A)$)]
We may assume that $X$ consists of a single element $e\in U$.
For otherwise, there exists a partition $\{Y,Z\}$ of $X$, and it follows that
$\proj(X, \mathcal{A})
 = \proj(Y\cup Z, \mathcal{A})
 = \proj(Y, \proj(Z, \mathcal{A}))$.
For any $\mathcal{C}\subseteq \Pi_k(U)\times\mathbb{N}$ and
 $q\in \Pi_k(U\setminus \{e\})$,
 it follows that
\[
\begin{split}
\opt_k(q,\proj(\{e\},\mathcal{C}))
& =\min\left\{w\mid\exists (p,w)\in\mathcal{C}:
 p_{\downarrow U\setminus\{e\}}\sqcup q\in\Gamma_k(U\setminus\{e\})
 \wedge \{e\}\notin p\right\}\\
& =\min\left\{w\mid\exists (p,w)\in\mathcal{C}: p\sqcup q_{\uparrow U}
 \in\Gamma_k(U)\right\}
 =\opt_k(q_{\uparrow U},\mathcal{C}).
\end{split}
\]
Therefore, $\proj$ preserves $k$-representation as
\[
\opt_k(q,\proj(\{e\},\mathcal{A}'))
 =\opt_k(q_{\uparrow U},\mathcal{A}')
 =\opt_k(q_{\uparrow U},\mathcal{A})
 =\opt_k(q, \proj(\{e\}, \mathcal{A})).
\]
\item[Join ($\join(\mathcal A,\mathcal B)$)]
We may assume
 $\mathcal{A},\mathcal{B} \subseteq \Pi_k(\hat{U})\times\mathbb{N}$.
For otherwise, it follows that
$\join(\mathcal{A}, \mathcal{B})
 = \join(\ins(U'\setminus U, \mathcal{A}), \ins(U\setminus U', \mathcal{B}))$.
For any $\mathcal{B}\subseteq\Pi_k(\hat{U})\times\mathbb{N}$ and
 $r\in\Pi_k(\hat{U})$, it follows that
\[
\begin{split}
\opt_k(r, \join(\mathcal{A}', \mathcal{B}))
& = \min\bigl\{w_a + w_b\mid\exists (p,w_a)\in\mathcal{A}':\exists (q,w_b)\in\mathcal{B}:
 p\sqcup q\sqcup r \in \Gamma_k(\hat{U}) \bigr\}\\
& =\min\left\{\opt_k(q\sqcup r, \mathcal{A}') + w_b \mid (q,w_b)\in\mathcal{B}\right\}\\
& =\min\left\{\opt_k(q\sqcup r, \mathcal{A}) + w_b \mid (q,w_b)\in\mathcal{B}\right\}\\
& =\min\bigl\{w_a + w_b \mid\exists (p,w_a)\in\mathcal{A}:\exists (q,w_b)\in\mathcal{B}: p\sqcup q\sqcup r \in \Gamma_k(\hat{U})\bigr\} \\
& =\opt_k(r, \join(\mathcal{A}, \mathcal{B})).
\end{split}
\]
Therefore, $\join$ preserves $k$-representation.
\end{description}
\end{proof}

We then present an algorithm \alg{$k$-reduce} that,
 for any $\mathcal{A}\subseteq \Pi_k(U)\times\mathbb{N}$,
 finds a small subset $\mathcal{A}'\subseteq \mathcal{A}$
 that $k$-represents $\mathcal{A}$.
The algorithm is a natural generalization of the algorithm presented for $k=1$
 in \cite{BCKN15}.
The following proposition will be used in our proof as well as in \cite{BCKN15}.
\begin{proposition}[\cite{BCKN15}]
\label{pr:Basis}
There exists an algorithm that,
 for an $n\times m$ matrix ($m\leq n$) on $\mathbb{F}_2$
 such that each row is weighted by a nonnegative integer,
 finds a basis of the row space with the minimum sum of weights
 in 
 $\mathcal{O}(nm^{\omega - 1})$ time, where
 $\omega$ is the matrix multiplication exponent.
\end{proposition}

\begin{lemma}[\alg{$k$-reduce}]
\label{lm:k-reduce}
Let $\mathcal{A}\subseteq\Pi_k(U)\times\mathbb{N}$.
If $k$ is even, then
 $\mathcal A'\subseteq\mathcal A$ that
 $k$-represents $\mathcal A$ and has a size at most $k^{|U|-k}$
 can be found in
 $|\mathcal A|k^{(\omega - 1)(|U| - k)}|U|^{\mathcal{O}(1)}$ time.
If $k$ is odd, then
 $\mathcal A'\subseteq\mathcal A$ that
 $k$-represents $\mathcal A$ and has a size at most $(k+1)^{|U|-k}$
 can be found in
 $|\mathcal A|(k + 1)^{(\omega - 1)(|U| - k)}|U|^{\mathcal{O}(1)}$ time.
\end{lemma}
\begin{proof}
The lemma is proved similarly to the lemmas for $k=1$ in \cite{BCKN15}.

We assume $k$ is even (the case for an odd $k$ will be discussed
 at the end of the proof),
 and
 define a $\Pi_k(U)\times\Pi_k(U)$ matrix $\mathcal{M}$
 and
 a $\Pi_k(U)\times \Gamma_k(U)$ matrix $\mathcal{C}$ on $\mathbb{F}_2$
 as follows:
\[
\mathcal{M}[p,q]=[p\sqcup q \in \Gamma_k(U)]\text{ and }
\mathcal{C}[p,c]=[c\sqsubseteq p],
\]
 where
 the brackets in the right hand sides
 denote the Iverson bracket, i.e.,
 for a proposition $P$,
 $[P]=1$ if $P$ is true, and $[P]=0$ if $P$ is false.
The weighted partition set $\mathcal A$ is associated
 with a subspace of the row space of the matrix $\mathcal C$, together with
 weights assigned to rows in the subspace.
Specifically, we suppose that
 $\mathcal A=\{(p_1,w_1),\ldots,(p_\ell,w_\ell)\}$ and
 $A=\{p_1,\ldots,p_\ell\}$.
Then, $\mathcal A$ can be regarded as
 the row space $\mathcal{C}[A,\cdot\,]$ such that each row
 $\mathcal C[p_i,\cdot\,]$ is weighted by $w_i$.
Let $X\subseteq A$ be a basis of the row space $\mathcal{C}[A,\cdot\,]$
 minimizing the sum of weights $\sum_{p_i\in X}w_i$, and
 $\mathcal A'=\{(p_i,w_i)\mid p_i\in X\}$.
We prove that
 $\mathcal A'$ $k$-represents $\mathcal A$
 through Claims \ref{cl:factorization} and~\ref{cl:k-representation}.
The size of $\mathcal A'$ and
 the running time for computing $\mathcal A'$
 are estimated in Claim~\ref{cl:Basis}.
\begin{claim}
\label{cl:factorization}
$\mathcal M$ can be factorized into $\mathcal C\mathcal C^\top$.
\end{claim}
\begin{claimproof}
It follows that
\[
(\mathcal{C}\mathcal{C}^\top)[p,q]
 =\sum_{c \in \Gamma_k(U)} [c \sqsubseteq p] \cdot [c \sqsubseteq q] \\
 =\sum_{c \in \Gamma_k(U)} [c \sqsubseteq p \wedge c \sqsubseteq q]\\
 =\sum_{c \in \Gamma_k(U)} [c \sqsubseteq p \sqcup q].
\]
Here, we count the number of partitions $c\in\Gamma_k(U)$ such that
 $p\sqcup q$ is a refinement of $c$.
If $p\sqcup q\notin\Pi_k(U)$, i.e.,
 if there is a block in $p\sqcup q$ containing $s_i, s_j$ for some $i\neq j$,
 then the number of such partitions $c$ is $0$.
If $p\sqcup q\in\Pi_k(U)$, then
 the block of a partition $c\sqsubseteq p\sqcup q$ containing $s_i$
 is a superset of the block of $p\sqcup q$ containing $s_i$
 and possibly of any block of $p\sqcup q$ containing none of
 elements $s_1,\ldots,s_k$.
Therefore, the number of such partitions $c$ is $k^{\nblock{p\sqcup q}-k}$.
If $p\sqcup q\in\Gamma_k(U)$, then
 because $\nblock{p\sqcup q} = k$, this number is $k^{k - k} = 1$.
If $p\sqcup q\in\Pi_k(U)\setminus\Gamma_k(U)$, then
 because $\nblock{p\sqcup q} > k$, this number is even for an even $k$.
We thus have
\[
\begin{split}
(\mathcal{C}\mathcal{C}^\top)[p,q]
&=\sum_{c \in \Gamma_k(U)} [c \sqsubseteq p \sqcup q]\\
&=
\begin{cases}
k^{\nblock{p\sqcup q} - k} & p \sqcup q\in\Pi_k(U)\\
0 & \text{otherwise}
\end{cases}\\
& \equiv
[p\sqcup q \in \Gamma_k(U)]\\
& = \mathcal{M}[p, q].
\end{split}
\]
\end{claimproof}
\begin{claim}
\label{cl:k-representation}
$\mathcal A'$ $k$-represents $\mathcal A$.
\end{claim}
\begin{claimproof}
We assume for contradiction that
 for some $q\in\Pi_k(U)$,
 $\opt_k(q,\mathcal A')\neq\opt_k(q,\mathcal A)$.
Since $\mathcal A'\subseteq\mathcal A$,
 it follows that $\opt_k(q,\mathcal A')>\opt_k(q,\mathcal A)$, which implies
 the existence of a weighted partition $(p,w)\in\mathcal A\setminus\mathcal A'$
 such that
 $p\sqcup q\in\Gamma_k(U)$ and $w<\opt_k(\mathcal A')$.
Since $X\subseteq A$
 is a basis of the row space $\mathcal C[A,\cdot\,]$,
 the row $\mathcal C[p,\cdot\,]$ can be represented as a linear combination of
 the rows of $\mathcal C[Y,\cdot\,]$ for some $Y\subseteq X$, i.e.,
 $\mathcal C[p,\cdot\,]\equiv\sum_{y\in Y}\mathcal C[y,\cdot\,]$.
Since
 $p\sqcup q\in\Gamma_k(U)$ and
 $\mathcal M=\mathcal C\mathcal C^\top$ by Claim~\ref{cl:factorization},
 it follows that
\[
1=\mathcal M[p,q]
\equiv\mathcal C[p,\cdot\,]\mathcal C^\top[\,\cdot,q]
\equiv\left(\sum_{y\in Y}\mathcal C[y,\cdot\,]\right)\mathcal C^\top[\,\cdot,q]\\
\equiv\sum_{y\in Y}\left(\mathcal C[y,\cdot\,]\mathcal C^\top[\,\cdot,q]\right)
\equiv\sum_{y\in Y}\mathcal M[y,q].
\]
Hence, there exists $p_i\in A$ such that
 $p_i\in Y$ and
 $\mathcal M[p_i,q]=1$, i.e., $p_i\sqcup q\in\Gamma_k(U)$.
Since $w<\opt_k(\mathcal A')$, $w_i>w$ follows.

The row $\mathcal C[p_i,\cdot\,]$ can also be represented as
 a linear combination of the rows of
 $\mathcal C[(Y\setminus\{p_i\})\cup\{p\},\cdot\,]$, since
 $\mathcal C[p,\cdot\,]\equiv\sum_{y\in Y}\mathcal C[y,\cdot\,]$.
Therefore,
 $(X\setminus\{p_i\})\cup\{p\}$ is also a basis of the row space
 $\mathcal C[A,\cdot\,]$.
This basis has the weight less than the weight of $X$ by $w_i-w$,
 which contradicts that
 $X$ is a basis with the minimum sum of weights.
\end{claimproof}
\begin{claim}
\label{cl:Basis}
$\mathcal A'$ has a size at most $k^{|U|-k}$ and
 can be found in
 $|\mathcal A|k^{(\omega - 1)(|U| - k)}|U|^{\mathcal{O}(1)}$ time.
\end{claim}
\begin{claimproof}
Because the rank of $\mathcal C$ is at most the number of columns,
 $|\Gamma_k(U)|$,
 the size of $\mathcal A'$ is at most $k^{|U|-k}$.
Hence, by Proposition~\ref{pr:Basis},
 $\mathcal A'$ can be obtained in
 $|\mathcal A|k^{(\omega-1)(|U|-k)}$ time.
\end{claimproof}

It remains to prove the lemma for the case that $k$ is odd.
For an odd $k$,
 we replace the definition of $\mathcal C$ with
 a $\Pi_k(U)\times\Gamma'_{k+1}(U)$ matrix
 $\mathcal C[p,q]=[c\sqsubseteq p]$, where
 $\Gamma'_{k+1}(U)$ is the set of partitions
 consisting of $k+1$ blocks $b_1$, \ldots, $b_{k+1}$ such that
 $s_i\in b_i$ for $1\leq i \leq k$, and that
 $b_{k+1}$ is exceptionally allowed to be an empty set.
As another exception, we also write
 $c\sqsubseteq r$ for a partition $r\in\Gamma_k(U)$ if
 $c$ contains an empty block and $r$ is a refinement of
 the partition consisting of the nonempty blocks of $c$.
We note that the definition of $\mathcal M$,
 $\mathcal M[p,q]=[p\sqcup q\in\Gamma_k(U)]$
 is not changed.

This replacement affects Claims \ref{cl:factorization} and~\ref{cl:Basis}.
In Claim~\ref{cl:factorization},
 the number of partitions $c\in\Gamma'_{k+1}(U)$ such that
 $c\sqsubseteq p\sqcup q$ is replaced with $(k+1)^{\nblock{p\sqcup q}-k}$.
This number is $(k+1)^{k-k} = 1$ if $p\sqcup q\in\Gamma_k(U)$,
 and an even number if $p\sqcup q\in\Pi_k(U)\setminus\Gamma_k(U)$.
Thus Claim~\ref{cl:factorization} holds for an odd $k$.
In Claim~\ref{cl:Basis},
 the size of $\mathcal A'$, i.e.,
 the rank of $\mathcal C$ is at most the number of columns
 $|\Gamma'_{k+1}(U)|=(k+1)^{|U|-k}$.
Thus, by Proposition~\ref{pr:Basis},
 the running time for computing $\mathcal A'$ is increased to
 $|\mathcal A|(k+1)^{(\omega-1)(|U|-k)}$.

Therefore, the proof of Lemma~\ref{lm:k-reduce} is completed.
\end{proof}

\subsection{Algorithm for \kDDP}
\label{ssc:k-Algorithm_tw}
The proposed algorithm first transforms a given tree decomposition
 of a given graph $G$
 into a nice tree decomposition $\mathcal T=(T,\{B_t\}_{t\in V(T)})$
 such that the root node has an empty bag
 using Proposition~\ref{pr:NiceTreeDecomposition}, and then
 adds all demand vertices
 $s_1$, \ldots, $s_k$, $t_1$, \ldots, $t_k$ to every bag.
Consequently,
 the width of $\mathcal T$ increases from $\tw$ to $\tw+2k$, and
 the root has a bag consisting only of the demand vertices.
By the addition of the demands,
 we can make a linear forest that is a partial solution of \kDDP{}
 contained in the graph $G_x=(V_x,E_x)$, which is
 induced by all the vertices and edges introduced by
 a node $x$ and its descendants,
 have demand vertices in the bag $B_x$ of the node $x$.
In particular, each connected component of the linear forest (which is a path)
 has at least one end-vertex in $B_x$.
Such a linear forest, $F_x$, is represented as a partition whose
 blocks consist of end-vertices, in $B_x$,
 of the connected components of $F_x$.
In addition, because the algorithm is based on the notions of
 $k$-representation and its preservation,
 we make every partition in the algorithm have at least $k$ blocks containing
 $s_1$, \ldots, $s_k$ separately.
The algorithm computes, by dynamic programming,
 the set of such partitions representing
 all linear forests $F_x$ that may be partial solutions for
 all possible combinations of indegree and outdegree
 ($(0,1)$ or $(1,0)$)
 of end-vertices, in $B_x$, of the connected components of $F_x$.
For computing the partition set, we only use operators
 in Definition~\ref{df:Operators},
 and then
 reduce the resulting partition set using \alg{$k$-reduce} in
 Lemma~\ref{lm:k-reduce}
 in order to achieve the desired running time.

In the following specific description,
 pairs $(i,o)$ of indegree $i$ and outdegree $o$ of a vertex frequently appear.
To avoid complication due to many parentheses,
 we denote the pair by $\delta_{io}$ and define the sum of pairs as
 $\delta_{io}+\delta_{i'o'}=\delta_{(i+i')(o+o')}$.
For a node $x$ and every
 $\mathbf{d}\in\{\delta_{00},\delta_{01},\delta_{10},\delta_{11}\}^{B_x}$,
 the partition set $A_x(\mathbf{d})$ representing
 all possible partial solutions is defined as follows.
\begin{align*}
A_x(\mathbf{d})
& = \biggl\{\left( p,\min_{X\in\mathcal{E}_x(p,\mathbf{d})}\omega(X)\right)\mid
 p\in\Pi_k\left(\mathbf{d}^{-1}(\delta_{01})\cup\mathbf{d}^{-1}(\delta_{10})
\cup\{s_1,\ldots,s_k\}\right)\\
&\qquad\qquad\qquad\qquad\qquad\quad\wedge
\mathcal{E}_x(p, \mathbf{d})\neq\emptyset
 \biggr\}\\
\begin{split}
\mathcal{E}_x(p, \mathbf{d})& =\bigl\{X \subseteq E_x \mid
 \forall v\in B_x: (\indeg{G[X]}(v), \outdeg{G[X]}(v)) = \mathbf{d}(v)\\
&\qquad\qquad\quad\wedge
 \forall v\in V_x\setminus B_x:
 (\indeg{G[X]}(v),\outdeg{G[X]}(v))\in\{\delta_{00},\delta_{11}\}\\
&\qquad\qquad\quad\wedge
 \text{$G[X]$ is a linear forest}\\
&\qquad\qquad\quad\wedge
\forall \{u,v\}\in p:
 \text{$G[X]$ has a connected component that is a $uv$-path}
\bigr\}
\end{split}
\end{align*}
The mapping $\mathbf{d}$ represents
 indegrees and outdegrees of vertices in the bag $B_x$ for a partial solution.
Any partition $p$ in $A_x(\mathbf{d})$ only contains
 vertices in $B_x$ that have
 degree $\delta_{01}$ or $\delta_{10}$
 in the partial solution,
 together with origins $s_1$, \ldots, $s_k$ in demands.
It is guaranteed by the definition of $\mathcal E_x$ that
 an edge set $X\in\mathcal E_x(p,\mathbf d)$
 induces a linear forest $G[X]$
 and that
 $p$ has blocks consisting of two end-vertices, in $B_x$,
 of the connected components of $G[X]$.
The weight of $X$ is denoted by $\omega(X)$.
Every time $A_x(\mathbf d)$ is computed, we reduce it to
 $A'_x(\mathbf d)=\text{\alg{$k$-reduce}}(A_x(\mathbf d))$.
The answer of \kDDP{} is \yes{} if,
 for the root node $r$ and
 $\mathbf{d}=\{
(s_1,\delta_{01}),(t_1,\delta_{10}),\ldots,(s_k,\delta_{01}),(t_k,\delta_{10})
\}$,
 there is a partition
 $\{\{s_1,t_1\},\ldots,\{s_k,t_k\}\}$ in $A'_r(\mathbf{d})$.
The minimum weight is obtained along with the partition.
Otherwise, the answer is \no.

The base and recursive steps 
 for computing $A_x(\mathbf{d})$ by dynamic programming
 is defined for each type of nodes as follows.

\subparagraph*{Leaf Node $x$}
We have the bag $B_x=\{s_1,t_1,\ldots,s_k,t_k\}$.
Because no edge is introduced, all the vertices in $B_x$ are isolated.
By the definition of $A_x(\mathbf{d})$,
 any partition in $A_x(\mathbf{d})$ needs to have each $s_i$ as
 the singleton block $\{s_i\}$.
We thus set as follows.
\[
A_x(\mathbf{d})=
\begin{cases}
\{( \{\{s_1\},\ldots,\{s_k\}\}, 0)\} 
& \forall v \in B_x:\mathbf{d}(v)=\delta_{00} \\
\emptyset & \text{otherwise}
\end{cases}
\]

\subparagraph*{Introduce Vertex $v$ Node $x$ with Child $y$}
We have the bag $B_x=B_y\cup\{v\}$.
Because no edge incident to $v$ is introduced, $v$ is isolated.
Therefore, for every
 $\mathbf{d}\in\{\delta_{00},\delta_{01},\delta_{10},\delta_{11}\}^{B_x}$,
 we can compute $A_x(\mathbf{d})$ as follows.
\[
A_x(\mathbf{d})=
\begin{cases}
A'_y(\mathbf{d}|_{B_y}) & \mathbf{d}(v)=\delta_{00}\\
\emptyset & \text{otherwise}
\end{cases}
\]

\subparagraph*{Forget Vertex $v$ Node $x$ with Child $y$}
We have the bag $B_x=B_y\setminus \{v\}$.
The vertex $v$ appears neither in $B_x$ nor in ancestor nodes of $x$.
Hence,
 for a linear forest $F$ in $G_x$ representing a partial solution,
 $v$ is either an internal vertex of a connected component of $F$
 or not contained in $F$.
Therefore, for every
 $\mathbf{d}\in\{\delta_{00},\delta_{01},\delta_{10},\delta_{11}\}^{B_x}$,
 we can compute $A_x(\mathbf{d})$ as follows.
\[
A_x(\mathbf{d}) = 
A'_y(\mathbf{d}[v\rightarrow \delta_{00}])
 \cuparrow A'_y(\mathbf{d}[v\rightarrow \delta_{11}])
\]
We note that since every bag has all demand vertices,
 these vertices are never forgotten.

\subparagraph*{Introduce Edge $uv$ Node $x$ with Child $y$}
We have the bag $B_x=B_y$.
For every
 $\mathbf{d}\in\{\delta_{00},\delta_{01},\delta_{10},\delta_{11}\}^{B_x}$,
 we compute $A_x(\mathbf{d})$ as follows.
\[
\begin{split}
A_x(\mathbf{d}) = A'_y(\mathbf{d})\cuparrow
&
\begin{cases}
\emptyset
& \mathbf{d}(u)=\delta_{00}\vee\mathbf{d}(v)=\delta_{00}\\
\gluew(uv,A'_y(\mathbf{d}[u\rightarrow\delta_{00},v\rightarrow\delta_{00}]))
& \mathbf{d}(u)=\delta_{01}\wedge\mathbf{d}(v)=\delta_{10}\\
\proj(\{u\},\gluew(uv,A'_y(\mathbf{d}[u\rightarrow\delta_{10},v\rightarrow\delta_{00}])))
& \mathbf{d}(u)=\delta_{11}\wedge\mathbf{d}(v)=\delta_{10}\\
\proj(\{v\},\gluew(uv,A'_y(\mathbf{d}[u\rightarrow\delta_{00},v\rightarrow\delta_{01}])))
& \mathbf{d}(u)=\delta_{01}\wedge\mathbf{d}(v)=\delta_{11}\\
\proj(\{u,v\},\gluew(uv,A'_y(\mathbf{d}[u\rightarrow\delta_{10},v\rightarrow\delta_{01}])))
& \mathbf{d}(u)=\delta_{11}\wedge\mathbf{d}(v)=\delta_{11}
\end{cases}
\end{split}
\]
Because the optimal solution may not contain the edge $uv$,
 we make $A_x(\mathbf{d})$ contain $A'_y(\mathbf{d})$.
If $\mathbf{d}(u)$ or $\mathbf{d}(v)$ equals $\delta_{00}$,
 then because no partial solution may contain the edge $uv$,
 we are done.
If the edge $uv$ is contained in a partial solution,
 then
 the outdegree of $u$ and the indegree of $v$ are both incremented by $1$.
If $\mathbf{d}(u)=\delta_{01}$ and $\mathbf{d}(v)=\delta_{10}$, then
 because the vertices $u$ and $v$ are isolated in $G_y$,
 we just add the block $\{u,v\}$ to
 partitions in
 $A'_y(\mathbf{d}[u\rightarrow\delta_{00},v\rightarrow\delta_{00}])$
 using $\gluew$.

If $\mathbf{d}(u)=\delta_{11}$ and $\mathbf{d}(v)=\delta_{10}$, then
 the vertex $u$ is a destination of a connected component
 of every linear forest constructed in $G_y$ for this case.
Therefore, 
 every partition in
 $A'_y(\mathbf{d}[u\rightarrow\delta_{10},v\rightarrow\delta_{00}])$
 has a two-element block containing $u$.
Suppose that $\{a,u\}$ is such a block.
In order to attach the edge $uv$ to the vertex $u$ and
 replace the block $\{a,u\}$ with $\{a,v\}$,
 we just need to add the vertex $v$ to $\{a,u\}$ ($\gluew$),
 and then to remove the vertex $u$ ($\proj$).
The case of
 $\mathbf{d}(u)=\delta_{01}$ and $\mathbf{d}(v)=\delta_{11}$ is similar.

If $\mathbf{d}(u)$ and $\mathbf{d}(v)$ are both $\delta_{11}$, then
 the vertices $u$ and $v$ are a destination and a origin, respectively,
 of one or two connected components of every linear forest constructed in
 $G_y$ for this case.
Therefore, each partition $p$ in
 $A'_y(\mathbf{d}[u\rightarrow\delta_{10},v\rightarrow\delta_{01}])$
 has either a block $\{v,u\}$ or
 two two-element blocks containing $u$ and $v$ separately, say,
 $\{a,u\}$ and $\{v,b\}$.
For the latter case,
 in order to add the edge $uv$ and replace the blocks
 $\{a,u\}$ and $\{v,b\}$ with $\{a,b\}$,
 we just need to join the blocks $\{a,u\}$, $\{v,b\}$
 with the vertices $u$, $v$ ($\gluew$),
 and then to remove the vertices $u$ and $v$ ($\proj$).
If
 the partition $p$ has the block $\{v,u\}$, then
 a cycle is created in the partial solution by adding the edge $uv$.
In this case, $p$ still has the block $\{v,u\}$ after the addition of $uv$.
Thus $p$ is removed when the vertices $u$, $v$ are removed ($\proj$),
 because the number of blocks is decremented.

\subparagraph*{Join Node $x$ with Children $y$ and $z$}
We have the bags $B_x=B_y=B_z$.
For every
 $\mathbf d_x\in\{\delta_{00},\delta_{01},\delta_{10},\delta_{11}\}^{B_x}$,
 we compute $A_x(\mathbf d_x)$ as follows.
\[
A_x(\mathbf{d}_x) =\!\bigcuparrow_{\mathbf{d}_y + \mathbf{d}_z = \mathbf{d}_x}\!
\proj\left(\mathbf{d}_x^{-1}(\delta_{11}) \setminus 
(\mathbf{d}_y^{-1}(\delta_{11}) \cup \mathbf{d}_z^{-1}(\delta_{11})),
\join(A'_y(\mathbf{d}_y), A'_z(\mathbf{d}_z))\right)
\]
We obtain a partition set representing linear forests in $G_x$
 as the union
 of the join of linear forests in $G_y$ and $G_z$,
 over all combinations of mappings $\textbf{d}_y$ and $\textbf{d}_z$
 satisfying $\textbf{d}_y+\textbf{d}_z=\textbf{d}_x$.
However,
 after the join of
 partitions $p_y$ and $p_z$,
 representing linear forests $F_y$ in $G_y$ and $F_z$ in $G_z$,
 respectively,
 the joined partition $p_y\sqcup p_z$ may have a block
 containing internal vertices of a path in the join of $F_y$ and $F_z$ or
 a block representing a cycle $C$ created in the join of $F_y$ and $F_z$.
We need to remove internal vertices from any block,
 and to remove the partition $p_y\sqcup p_z$ from the resulting partition set
 if it contains a block representing a cycle.

If a block $b$ in $p_y\sqcup p_z$ contains a cycle $C$, then
 because $C$ is the union of some connected components (paths)
 in $F_y$ and $F_z$,
 the block $b$ consists of the vertices in $C$ that are contained in $B_x$.
All the vertices of $b$ have degree $\delta_{11}$ in $C$ and
 do not have degree $\delta_{11}$ in $F_y$ or $F_z$,
 because they are end-vertices of paths in $F_y$ or $F_z$.
Therefore,
 the partition $p_y\sqcup p_z$ is removed
 by removing
 $\mathbf{d}_x^{-1}(\delta_{11}) \setminus
 (\mathbf{d}_y^{-1}(\delta_{11}) \cup \mathbf{d}_z^{-1}(\delta_{11}))$
 (i.e., the vertices with degree newly incremented to $\delta_{11}$)
 using $\proj$,
 since $b$ is removed by this operation, which decrements the number of blocks.
Besides, this removal operation just removes internal vertices in
 any block that does not represent a cycle.

By Lemma~\ref{lm:Representation},
 operations for computing $A_x$ in each node
 preserve $k$-representation.
We complete the proof of Theorem~\ref{th:Algorithm_tw} by
 estimating the time complexity of the proposed algorithm.
\begin{lemma}
\label{lm:Algorithm_tw_Time}
The proposed algorithm solves \kDDP{} in
 $2^{\omega\cdot (\tw + 2k)\log k} \cdot 2^{\mathcal{O}(\tw + k)}\cdot n$ time,
 where
 $\omega$ is the matrix multiplication exponent.
\end{lemma}
\begin{proof}
We estimate the running time of \alg{$k$-reduce} applied to a join node,
 because it is the bottleneck in the dynamic programming.
For a join node $x$ with children $y$ and $z$ and
 $\mathbf{d}_x\in\{\delta_{00},\delta_{01},\delta_{10},\delta_{11}\}^{B_x}$,
 by Lemma~\ref{lm:k-reduce},
 $A'_y(\mathbf d_y)$ and $A'_z(\mathbf d_z)$
 have sizes at most
 $(k+1)^{|\mathbf{d}_y^{-1}(\delta_{01})\cup\mathbf{d}_y^{-1}(\delta_{10})|+k-k}$
 and
 $(k+1)^{|\mathbf{d}_z^{-1}(\delta_{01})\cup\mathbf{d}_z^{-1}(\delta_{10})|+k-k}$,
 respectively.
Therefore, $A_x(\mathbf{d}_x)$ has a size at most
\[
\begin{split}
& \sum_{\mathbf{d}_y + \mathbf{d}_z = \mathbf{d}_x}
(k + 1)^{|\mathbf{d}_y^{-1}(\delta_{01}) \cup \mathbf{d}_y^{-1}(\delta_{10})|}
 \cdot
(k + 1)^{|\mathbf{d}_z^{-1}(\delta_{01}) \cup \mathbf{d}_z^{-1}(\delta_{10})|} \\
&= \prod_{v\in B_x} \sum_{\mathbf d_y(v) + \mathbf d_z(v) = \mathbf d_x(v)} 
(k + 1)^{[\mathbf d_y(v) = \delta_{01}] + [\mathbf d_y(v) = \delta_{10}] + [\mathbf d_z(v) = \delta_{01}] + [\mathbf d_z(v) = \delta_{10}]} \\
&=\{2(k + 1)\}^{|\mathbf{d}_x^{-1}(\delta_{01}) \cup \mathbf{d}_x^{-1}(\delta_{10})|} \cdot
\{2(k + 1)^2 + 2\}^{|\mathbf{d}_x^{-1}(\delta_{11})|},
\end{split}
\]
where
 the last equality is derived using
\[
\begin{split}
&\sum_{\mathbf d_y(v)+\mathbf d_z(v)=\mathbf d_x(v)} 
(k+1)^{[\mathbf d_y(v)=\delta_{01}]+[\mathbf d_y(v)=\delta_{10}]+[\mathbf d_z(v)=\delta_{01}]+[\mathbf d_z(v)=\delta_{10}]}\\
&\qquad = 
\begin{cases}
1 & \mathbf d_x(v)=\delta_{00}\\
2(k+1) &  \mathbf d_x(v)\in\{\delta_{01},\delta_{10}\}\\
2(k+1)^2+2 & \mathbf d_x(v)=\delta_{11}.
\end{cases}
\end{split}
\]
By Lemma~\ref{lm:k-reduce},
 the running time for applying \alg{$k$-reduce} to $A_x(\mathbf{d}_x)$ is
\[
\begin{split}
&\{2(k + 1)\}^{|\mathbf{d}_x^{-1}(\delta_{01}) \cup \mathbf{d}_x^{-1}(\delta_{10})|} \cdot
\{2(k + 1)^2 + 2\}^{|\mathbf{d}_x^{-1}(\delta_{11})|}
\cdot (k + 1)^{(\omega - 1)|\mathbf{d}_x^{-1}(\delta_{01}) \cup \mathbf{d}_x^{-1}(\delta_{10})|}
\cdot (\tw + k)^{\mathcal{O}(1)} \\
&\leq\{2(k + 1)^\omega\}^{|\mathbf{d}_x^{-1}(\delta_{01}) \cup \mathbf{d}_x^{-1}(\delta_{10})|} \cdot
\{2(k + 1)^2 + 2\}^{|\mathbf{d}_x^{-1}(\delta_{11})|}
\cdot (\tw + k)^{\mathcal{O}(1)}.
\end{split}
\]
Because we need to apply \alg{$k$-reduce} to $A_x(\mathbf{d}_x)$ for all
 $\mathbf{d}_x\in\{\delta_{00},\delta_{01},\delta_{10},\delta_{11}\}^{B_x}$,
 this total running time is at most
\[
\begin{split}
&\sum_{\mathbf{d}_x \in \{\delta_{00}, \delta_{01}, \delta_{10}, \delta_{11}\}^{B_x}}
 \{2(k + 1)^\omega\}^{|\mathbf{d}_x^{-1}(\delta_{01}) \cup \mathbf{d}_x^{-1}(\delta_{10})|}
 \cdot\{2(k+1)^2+2\}^{|\mathbf{d}_x^{-1}(\delta_{11})|}
 \cdot (\tw+k)^{\mathcal{O}(1)}\\
&=\sum_{i_{00}+i_{01}+i_{10}+i_{11}=|B_x|}
\binom{|B_x|}{i_{00},i_{01},i_{10},i_{11}}1^{i_{00}}
 \cdot \{2(k+1)^\omega\}^{i_{01}+i_{10}}
 \cdot\{2(k+1)^2+2\}^{i_{11}}
 \cdot (\tw+k)^{\mathcal{O}(1)}\\
&\leq \{4(k+1)^\omega+2(k+1)^2+3\}^{\tw+2k+1}\cdot (\tw+k)^{\mathcal{O}(1)}\\
&=k^{\omega\cdot (\tw+2k+1)}\cdot 2^{\mathcal{O}(\tw+k)}\\
&=2^{\omega\cdot (\tw+2k)\log k}\cdot 2^{\mathcal{O}(\tw+k)}.
\end{split}
\]
Because the number of recursive steps is $\mathcal O(n)$ by
 Proposition~\ref{pr:NiceTreeDecomposition},
 we obtain the lemma.
\end{proof}

By Lemmas \ref{lm:Representation}--\ref{lm:Algorithm_tw_Time},
 the proof of Theorem~\ref{th:Algorithm_tw} is completed.

%% file: section4.tex
\section{Tight Lower Bound under the SETH for General $k$}
\label{sc:LowerBound_SETH}
We prove Theorem~\ref{th:LowerBound_SETH}
 in Section~\ref{ssc:LowerBound_SETH}.
For a CNF formula with $n$ variables, an input of \CSAT,
 and for certain positive integers $\alpha$ and $K$ determined by $n$,
 we transform the CNF formula into an input of \aKKPHS,
 defined below, and then transform it into an input of \kDDP.
%
\begin{oframed}
\noindent\aKKPHS
\begin{description}
\item[Input]
Nonempty sets $S_1$, \ldots, $S_m\subseteq [\alpha K]\times [K]$.
\item[Question]
Is there a set $S$
 such that
 $S$ contains exactly one element from each row and each column
 of the grid on $\{(q-1)K+1,\ldots,qK\}\times [K]$ for each $1\leq q\leq\alpha$,
 and that
 $S\cap S_p\neq\emptyset$ for each $1\leq p\leq m$?
\end{description}
\end{oframed}
Transformation of the CNF formula into an input of the non-permutation version
 of \aKKPHS{}, where the condition of choosing one element from each column
 is removed, is simple.
To transform the CNF formula into the permutation version,
 we use the technique of \cite{LMS18}.
This technique was used to prove that
 \KKClique, where
 for a given graph with $K^2$ vertices placed on $[K]\times [K]$,
 we are asked the existence of a clique containing
 exactly one vertex from each row,
 can be reduced to
 the permutation version, i.e.,
 the problem of asking the existence of a clique containing
 exactly one vertex from each row and each column,
 in $2^{\mathcal O(K)}$ time.
Actually, this technique depends on no properties of the graph and its clique
 at all,
 and is realized just through
 shuffling elements of each row in the grid on $[K]\times [K]$.
Therefore, this technique can be generalized as follows.
\begin{theorem}[\cite{LMS18}]
\label{th:Permutation}
%
Suppose that $K^2$ elements are placed on $[K]\times [K]$ and
 that a certain condition
 for
 $K$ elements
 is
 satisfied by choosing exactly one element from each row.
Then, we can shuffle the elements in every row so that
 the same condition is satisfied by choosing exactly one element
 from each row and each column,
 in $2^{\mathcal O(K)}$ time.
%
\end{theorem}

For the reduction from \aKKPHS{} to \kDDP{},
 we partially use the reduction from \KKPHS{} to \kDDP{} in \cite{LMS18}.
More specifically,
 we divide the input of \aKKPHS{}
 $S_1, \ldots, S_m\subseteq [\alpha K]\times [K]$, transformed from the CNF formula,
 into
 $\alpha$ inputs of \KKPHS{}
 $S^q_1, \ldots, S^q_m\subseteq [K]\times [K]$ ($1\leq q\leq \alpha$),
 transform these $\alpha$ inputs into $\alpha$ inputs (graphs and demand sets)
 of \kDDP{} as in \cite{LMS18},
 and then put these graphs and demand sets together with slight modifications.

In Section~\ref{ssc:Algorithm_pw},
 we demonstrate that
 \kDDP{} can be solved in a running time asymptotically matching
 the lower bound of Theorem~\ref{th:LowerBound_SETH} using
 a modified version of the algorithm of \cite{Scheffler94}
 designed for \kUDP.

\subsection{Lower Bound under the SETH}
\label{ssc:LowerBound_SETH}
\subsubsection{Reduction from \CSAT{} to \aKKPHS}
Let $\varphi$ be an input CNF formula of \CSAT{} with $n$ variables
 $x_1$, \ldots, $x_n$ and $m$ clauses $C_1$, \ldots, $C_m$.
We define
 $X=\lfloor\log\frac n{\log^2 n}\rfloor$,
 $K=2^X$, and
 $\alpha=\lceil\lceil n/X\rceil/2^X\rceil=\lceil\lceil n/X\rceil/K\rceil$.
It follows that
\begin{align}
\label{eq:2XXUpper}
2^XX
&\leq\frac n{\log^2 n}\cdot\log\frac n{\log^2 n}
\leq\frac n{\log n}=o(n),\\
\label{eq:2XXLower}
2^XX
&\geq\frac n{2\log^2 n}\cdot\log\frac n{2\log^2 n}
=\Omega\left(\frac n{\log n}\right),\text{ and}\\
\label{eq:a2XX}
\alpha 2^XX
& \leq\left\lceil\frac{\lceil n/X\rceil}{2^X}\right\rceil\cdot 2^XX
\leq n+X+2^XX
=n+o(n)\text{ by (\ref{eq:2XXUpper}).}
\end{align}
We note $\alpha KX\geq n$.
If $\alpha KX>n$, then
 we add dummy variables $x_{n+1}$, \ldots, $x_{\alpha KX}$ to $\varphi$.
These additional variables do not affect the satisfiability of $\varphi$,
 because they appear in no clauses.
We divide all the variables into $\alpha K$ sets
 each containing $X$ variables,
 i.e.,
 $V_i=\{x_{(i-1)X+1},\ldots,x_{iX}\}$ ($1\leq i\leq\alpha K$).
For each $1\leq i\leq\alpha K$,
 let
 $\sigma_i:[K]\rightarrow\{1,0\}^{V_i}$
 be a one-to-one mapping
 from the positive integers in $[2^{|V_i|}]=[K]$ to
 all the truth assignments of the variables in $V_i$.
A typical definition of $\sigma_i$ is that for each $1\leq j\leq K=2^X$,
 $\sigma_i(j)$ assigns the variable $x_{(i-1)X+h}$ ($1\leq h\leq X$)
 the $h$th bit of the binary representation of $j-1$.
We construct an input $S'_1$, \ldots, $S'_m$ of \aKKHS, i.e.,
 the non-permutation version of \aKKPHS{} such that
 the condition of choosing one element from each column is removed:
For each $1\leq i\leq \alpha K$ and $1\leq j\leq K$,
 if $\sigma_i(j)$ satisfies a clause $C_p$, then
 we add an element $(i,j)$ to $S'_p$, i.e.,
 we define $S'_p=\{(i,j)\mid\text{$\sigma_i(j)$ satisfies $C_p$}\}$
 ($1\leq p\leq m$).
\begin{lemma}
\label{lm:CSAT-aKKHS}
The CNF formula $\varphi$ is satisfiable if and only if
 a solution of \aKKHS{} for the input $S'_1$, \ldots, $S'_m$ exists.
\end{lemma}
\begin{proof}
We first assume that $\varphi$ is satisfiable, i.e.,
 there is a truth assignment $\sigma_i(j_i)$
 for each $1\leq i\leq\alpha K$ such that
 $\varphi$ is satisfied.
Let $S=\{(i,j_i)\mid 1\leq i\leq\alpha K\}$.
Obviously, $S$ contains exactly one element from each row.
Because, for each $1\leq p\leq m$, the clause $C_p$ is satisfied by
 $\sigma_i(j_i)$ for some $i$,
 it follows that $S\cap S'_p\neq\emptyset$.
Therefore, $S$ is a solution of \aKKHS{} for $S'_1$, \ldots, $S'_m$.

We next assume that
 a solution of \aKKHS{} for the input $S'_1$, \ldots, $S'_m$ exists, i.e.,
 there exists $S=\{(i,j_i)\mid 1\leq i\leq\alpha K\}$ such that
 $S\cap S'_p\neq\emptyset$ for each $1\leq p\leq m$.
For each $1\leq p\leq m$ and any element $(i,j_i)$ in $S\cap S'_p\neq\emptyset$,
 by the definition of $S'_p$, the clause $C_p$ is satisfied by
 the truth assignment $\sigma_i(j_i)$.
Therefore, $\varphi$ is satisfiable.
\end{proof}
\begin{lemma}
\label{lm:CSAT-aKKPHS}
For the defined integers $\alpha$ and $K$,
\CSAT{} can be reduced to \aKKPHS{} in $2^{o(n)}+o(n^2 m)$ time.
\end{lemma}
\begin{proof}
We reduce \CSAT{} to \aKKHS{} using Lemma~\ref{lm:CSAT-aKKHS}.
We can do this by,
 for a variable $x_{(i-1)X+h}$ appeared in $C_p$,
 adding $(i,j)$ to $S'_p$
 if the variable appears in $C_p$ without negation and
 is assigned $1$ by $\sigma_i(j)$,
 or
 if the variable appears in $C_p$ with negation and
 is assigned $0$ by $\sigma_i(j)$.
Because, for each $1\leq h\leq X$,
 there are $2^{X-1}$ $j$'s such that $\sigma_i(j)$ assigns
 the variable $x_{(i-1)X+h}$ either $1$ or $0$,
 it takes
 $\mathcal O(2^{X-1}\cdot nm)\leq\mathcal O(\frac{n}{\log^2 n}\cdot nm)
=o(n^2m)$ time.

We then reduce to \aKKPHS{} using Theorem~\ref{th:Permutation} for
 each of $\alpha$ grids on $[K]\times [K]$,
 obtained by dividing the underlying grid on $[\alpha K]\times [K]$, which takes
 $\alpha\cdot 2^{\mathcal O(K)}
\leq n\cdot 2^{\mathcal O(n/\log^2 n)}=2^{o(n)}$.

We thus obtain the total reduction time of $o(n^2m)+2^{o(n)}$.
\end{proof}

\subsubsection{Reduction to \kDDP}
We transform the input $\mathcal I$ of \aKKPHS{} consisting of
 $S_1, \ldots, S_m\subseteq [\alpha K]\times[K]$,
 transformed from the CNF formula $\varphi$ using Lemma~\ref{lm:CSAT-aKKPHS},
 into an input $\bar{\mathcal I}$ of \kDDP.
For each $1\leq p\leq m$ and $1\leq q\leq\alpha$,
 let $S^q_p=S_p\cap M_q$, where
 $M_1$, \ldots, $M_\alpha$ are the $\alpha$ grids on $[K]\times [K]$
 obtained by dividing $[\alpha K]\times [K]$ on which
 $S_1, \ldots, S_m$ are defined.
By this division of $\mathcal I$,
 we obtain $\alpha$ inputs $\mathcal I_1$, \ldots, $\mathcal I_\alpha$
 of \KKPHS{}, consisting of
 $S^q_1, \ldots, S^q_m\subseteq [K]\times [K]$
 for each $1\leq q\leq\alpha$.
Here, we allow $S^q_p$ to be empty and not to be hit.

We transform
 $\mathcal I_1$, \ldots, $\mathcal I_\alpha$
 into
 $\alpha$ inputs (graphs and demand sets)
 $\mathcal I'_1$, \ldots, $\mathcal I'_\alpha$ of
 \kDDP{} as in \cite{LMS18}, and
 obtain one input (graph and demand set) $\mathcal I'$ by
 taking the disjoint union of these graphs and demand sets.
By the result of \cite{LMS18},
 for each $1\leq q\leq\alpha$,
 a solution for $\mathcal I_q$ exists if and only if
 a solution for $\mathcal I'_q$ exists.
Besides, if a solution for $\mathcal I'$ exists, then
 a solution for $\mathcal I$ exists.
This is because
 a solution for $\mathcal I'$ guarantees that for each $p$,
 all the nonempty sets in $S^1_p$, \ldots, $S^\alpha_p$ are hit,
 and hence $S_p$ is hit.
However, the converse is not necessarily true.
This is because
 a solution for $\mathcal I$ implies that
 $S_p$ may be hit in at least one of the sets $S^1_p$, \ldots, $S^\alpha_p$,
 which does not guarantee the existence of disjoint paths for
 all graphs and demand sets
 $\mathcal I'_1$, \ldots, $\mathcal I'_\alpha$.
To resolve this inequivalence, we modify $\mathcal I'$ in such a way that
 hitting at least one of the sets $S^1_p$, \ldots, $S^\alpha_p$
 implies the existence of disjoint paths for the modified graph and demands.

At this point, we avoid to describe the definition of $\mathcal I'$ in detail
 and only mention properties shown in \cite{LMS18} that are
 necessary for the modification to $\mathcal I'$
 and the correctness of the resulting reduction.
The specific construction of the input graph of \kDDP{} will be provided
 later in order to estimate the pathwidth.

\begin{property}[Graphs and demand sets of
 $\mathcal I'_1$, \ldots, $\mathcal I'_\alpha$]
Let $1\leq q\leq\alpha$ and
 $G_q$ be the graph of $\mathcal I'_q$.
\begin{itemize}
\item
For every $(i,j)\in [K]\times [K]$, $G_q$ contains a certain vertex $d^q_{i,j}$.
\item
For every $1\leq p\leq m$, $G_q$ contains certain vertices $s^q_p$ and $t^q_p$.
\item
For every $1\leq p\leq m$ and $(i,j)\in [K]\times [K]$
 such that $(i,j)\in S^q_p$,
 $G_q$ contains edges $(s^q_p,d^q_{i,j})$ and $(d^q_{i,j},t^q_p)$.
The vertices $s^q_p$ and $t^q_p$ are only incident to these edges.
\item
The demand set of $\mathcal I'_q$ contains $(s^q_p,t^q_p)$
 unless $S^q_p=\emptyset$.
\end{itemize}
\end{property}
%
\begin{lemma}[\cite{LMS18}]
\label{lm:Gq}
For each $1\leq q\leq\alpha$, the following are equivalent.
\begin{itemize}
\item
There exists a set $S^q=\{(i,\rho_q(i))\mid i\in [K]\}$, where
 $\rho_q$ is a permutation $[K]\rightarrow [K]$,
 such that
 $S^q$ hits all the nonempty sets in $S^1_p$, \ldots, $S^\alpha_p$, i.e.,
 $S^q\cap S^q_p\neq\emptyset$
 for every $1\leq p\leq m$ with $S^q_p\neq\emptyset$.
\item
The graph $G_q$ has disjoint paths for the demand set of $\mathcal I'_q$,
 one of which is a path $(s^q_p,d^q_{i,\rho_q(i)},t^q_p)$ for some $i$.
\end{itemize}
\end{lemma}

We modify $\mathcal I'$ to
 an input of \kDDP, denoted by $\bar{\mathcal I}$, as follows.
\begin{itemize}
\item
The vertices $s^1_p$, \ldots, $s^\alpha_p$ are identified with
 a single vertex $s_p$.
\item
The vertices $t^1_p$, \ldots, $t^\alpha_p$ are identified with
 a single vertex $t_p$.
\item
The demands $(s^1_p,t^1_p)$, \ldots, $(s^\alpha_p,t^\alpha_p)$
 are identified with a single demand $(s_p,t_p)$.
\end{itemize}
\begin{lemma}
\label{lm:aKKPHS-kDDP}
A solution for the input $\mathcal I$ of \aKKPHS{} exists
 if and only if
 a solution for the input $\bar{\mathcal I}$ of  \kDDP{} exists.
\end{lemma}
\begin{proof}
We first assume that a solution for $\mathcal I$ exists.
This means that for each $1\leq q\leq\alpha$,
 there exists a set $S^q=\{(i,\rho_q(i))\mid i\in [K]\}$, together with
 a permutation $\rho_q:[K]\rightarrow [K]$, such that
 for every $1\leq p\leq m$,
 there exists $1\leq q_p\leq\alpha$ with $S^{q_p}\cap S^{q_p}_p\neq\emptyset$.
Let $\mathcal I''_1$, \ldots, $\mathcal I''_\alpha$ be $\alpha$ inputs of \kDDP{}
 obtained from $\mathcal I'_1$, \ldots, $\mathcal I'_\alpha$
 by,
 for every $1\leq p\leq m$,
 removing $(s^q_p$, $t^q_p)$ from every $\mathcal I'_q$ with $q\neq q_p$.
By lemma~\ref{lm:Gq},
 for each $1\leq q\leq\alpha$,
 the graph $G_q$ has disjoint paths connecting the demands of $\mathcal I''_q$.
Putting these paths together, we obtain disjoint paths connecting the demands of
 $\bar{\mathcal I}$.

We next assume that a solution for $\bar{\mathcal I}$ exists.
This means that for every $1\leq p\leq m$,
 there exists $1\leq q_p\leq\alpha$ such that
 one of the disjoint paths in the solution
 connecting the demand $(s_p,t_p)$
 passes through the vertex $d^{q_p}_{i,j}$ for some $i$ and $j$.
Defining
 $\mathcal I''_1$, \ldots, $\mathcal I''_\alpha$ as in the former case,
 the disjoint paths in the solution for $\bar{\mathcal I}$
 form solutions for
 $\mathcal I''_1$, \ldots, $\mathcal I''_\alpha$.
By Lemma~\ref{lm:Gq}, for each $1\leq q\leq\alpha$,
 there exist a set $S^{q}=\{(i,\rho_{q}(i))\mid i\in [K]\}$ and
 a permutation $\rho_{q}:[K]\rightarrow [K]$ such that
 for every $1\leq p\leq m$,
 if one of the disjoint paths connecting the demand $(s_p,t_p)$
 passes through the vertex $d^{q_p}_{i,j}$, then $\rho_{q_p}(i)=j$,
 and such that
 for every $1\leq p\leq m$ with $S^{q}_p\neq\emptyset$,
 $S^{q}\cap S^{q}_p\neq\emptyset$.
We thus obtain a solution $S=\bigcup_{1\leq q\leq\alpha}S^q$ for $\mathcal I$.
\end{proof}

\subsubsection{Estimation of Pathwidth}
We describe the definition of the graph $G$ of $\bar{\mathcal I}$
 in order to estimate the pathwidth $\pw$ of $G$.
The graph $G$ is composed of the graphs $G_1$, \ldots, $G_\alpha$
 of $\mathcal I'_1$, \ldots, $\mathcal I'_\alpha$, respectively,
 defined in \cite{LMS18}.
For each $1\leq q\leq \alpha$,
 $G_q$ is composed of a collection of gadgets $G[p,q]$ ($1\leq p\leq m$),
 consisting of the following components.
\begin{itemize}
\item
A vertex $a_i$ for every $1\leq i\leq K$.
\item
A vertex $b_j$ for every $1\leq j\leq K$.
\item
A vertex $v_{i, j}$ and edges $(a_i,v_{i, j})$ and $(v_{i, j}, b_j)$
 for every $1\leq i\leq K$ and $1\leq j\leq K$.
\item
A path
 $P_i=(c_{i,0},d_{i,1},v^*_{i,1},c_{i,1},\ldots,d_{i,K},v^*_{i,K},c_{i,K})$
 for every $1\leq i\leq K$.
\item
Vertices $f_{i,j}$, $f^1_{i,j}$, and $f^2_{i,j}$, and
edges $(f^1_{i,j},f_{i,j})$, $(f_{i,j},f^2_{i,j})$, $(b_j, f_{i, j})$,
 $(f_{i, j}, c_{i, j})$, $(f^1_{i, j}, c_{i, 0})$, and
 $(c_{i, j - 1}, f^2_{i, j})$
 for every $1\leq i\leq K$ and $1\leq j\leq K$.
\item
Vertices $s_p, t_p$, and
 for every $i$ and $j$ satisfying $(i,j)\in S^q_p$,
 edges $(s_p,d_{i,j})$ and $(d_{i,j},t_p)$.
\end{itemize}

We obtain the graph $G_q$ by,
 for every $1\leq p< m$, $1\leq i\leq K$, and $1\leq j\leq K$,
 identifying
 the vertex $v^*_{i, j}$ in the gadget $G[p, q]$ and
 the vertex $v_{i, j}$ in the gadget $G[p+1, q]$.
As already mentioned,
 we obtain the graph $G$ by identifying the vertices
 $s_p$ ($t_p$, resp.) contained in $G_1$, \ldots, $G_\alpha$ with
 a single vertex $s_p$ ($t_p$, resp.) (Fig.~\ref{fig:gadget}).
\begin{figure}[t]
\centering
\raisebox{-\height}{\includegraphics[scale=0.85]{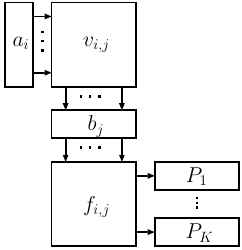}}
\quad
\raisebox{-\height}{\includegraphics[scale=0.75]{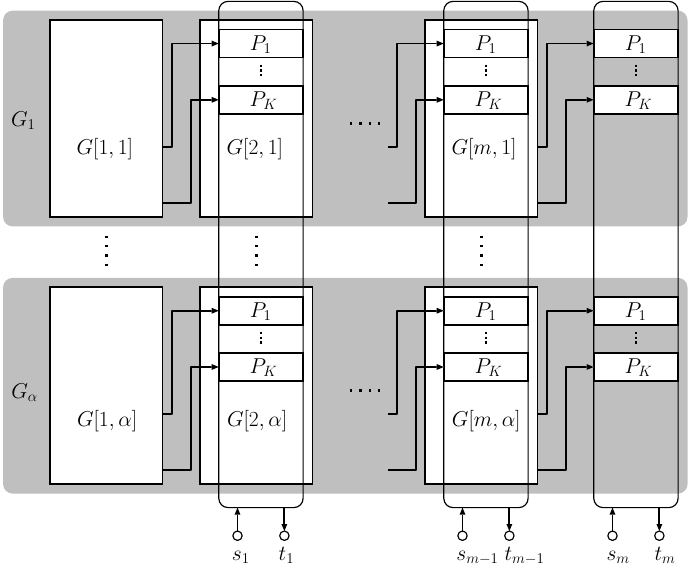}}
\caption{Gadget $G[p,q]$ (left) and the input graph $G$ of \kDDP{} (right).}
\label{fig:gadget}
\end{figure}
\begin{lemma}
\label{lm:pw}
$G$ has the pathwidth $\pw=\alpha K+\mathcal{O}(K)$.
\end{lemma}
\begin{proof}
We construct the path decomposition of $G$.
For every $0\leq p\leq m$, $1\leq q\leq \alpha$, $1\leq i\leq K$, and
 $1\leq j\leq K$,
 let $B[p,q,i,j]$ be a bag containing the following vertices.
\begin{itemize}
\item
The vertices $f_{i, j}$, $f^1_{i, j}$, $f^2_{i, j}$, $s_p$, $t_p$,
 and the vertices of $P_i$ in $G[p,q]$ (if $p>0$).
\item
The vertex $a_i$ in $G[p+1,q]$ (if $p<m$).
\item
The vertices $b_1$, \ldots, $b_K$ in
$G[p,q]$, $G[p,q+1]$, \dots, $G[p,\alpha]$ (if $p>0$).
\item
The vertices $b_1$, \dots, $b_K$ in $G[p+1,1]$, \dots, $G[p+1,q]$ (if $p<m$).
\end{itemize}
In addition,
 we add the vertices $v_{i, 1}$, \ldots, $v_{i, K}$ in $G[1,q]$ to $B[0,q,i,j]$.
We can obtain a path decomposition of $G$ by ordering these bags
 lexicographically according to $(p,q,i,j)$.
Actually, we can observe the following for every
 $1\leq p\leq m$ and $1\leq q\leq\alpha$.
\begin{itemize}
\item
The vertices $f_{i,j}$, $f^1_{i,j}$, and $f^2_{i,j}$ in $G[p,q]$
 appear only in the bag $B[p,q,i,j]$.
\item
The vertices $s_p$ and $t_p$ in $G[p,q]$
 appear only in the consecutive bags from $B[p,1,1,1]$ to $B[p,\alpha,K,K]$.
\item
The vertices of the path $P_i$ in $G[p,q]$ appear only in
 the consecutive bags from $B[p,q,i,1]$ to $B[p,q,i,K]$.
\item
The vertex $a_i$ in $G[p,q]$ appear only in
 the consecutive bags from $B[p-1,q,i,1]$ to $B[p-1,q,i,K]$.
\item
The vertices $b_1$, \ldots, $b_K$ in $G[p, q]$ appear only in
 the consecutive bags from $B[p-1,q,1,1]$ to $B[p,q,K,K]$.
\item
The vertices $v_{i,1}$, \ldots, $v_{i,K}$ in $G[1,q]$ appear only in
 the consecutive bags $B[0,q,i,1]$ to $B[0,q,i,K]$.
\end{itemize}
From these observations, we can also observe that
 for every edge in $G$, there is a bag containing
 both of the end-vertices of the edge.
Therefore, $G$ has the pathwidth
 $\pw\leq \{3+2+(3K+1)+1+(\alpha+1)K\}-1
 =\alpha K+\mathcal{O}(K)$.
\end{proof}

We complete the proof of Theorem~\ref{th:LowerBound_SETH} by
 deriving $\pw\log\pw=n+o(n)$ for the number $n$ of variables
 of the input CNF formula $\varphi$ of \CSAT.
%
It follows from (\ref{eq:2XXLower}) and the definition of $X$ that
\begin{equation}
\label{eq:loga}
\log\alpha=\log\left(\mathcal O\left(\frac n {2^XX}\right)\right)
=\mathcal O(\log\log n)=o(X).
\end{equation}
It follows from Lemma~\ref{lm:pw},
 (\ref{eq:2XXUpper}), (\ref{eq:a2XX}), and (\ref{eq:loga}) that
\begin{alignat*}{2}
\pw\log\pw
&=(\alpha K+\mathcal O(K))\log(\alpha K+\mathcal O(K))
& \qquad & \text{[by Lemma~\ref{lm:pw}]}\\
&=(\alpha K+\mathcal O(K))(\log\alpha+\log K+\mathcal O(1))\\
&=\alpha 2^X(o(X)+X)+\mathcal O(2^X)(o(X)+X))
& \qquad &\text{[by (\ref{eq:loga})]}\\
&=n+o(n).
& \qquad &\text{[by (\ref{eq:a2XX}) and (\ref{eq:2XXUpper})]}
\end{alignat*}
Combined with Lemmas \ref{lm:CSAT-aKKPHS} and~\ref{lm:aKKPHS-kDDP},
 Theorem~\ref{th:LowerBound_SETH} is obtained.

\subsection{Tight Algorithm for General $k$}
\label{ssc:Algorithm_pw}
In order to modify the algorithm designed for \kUDP{} in \cite{Scheffler94}
 so that it can solve \kDDP,
 we describe properties necessary for the modification.\footnote{%
The following description is a reinterpretation of \cite{Scheffler94}
 and differs from the original description in some details.}
Suppose that an $n$-vertex undirected graph $G$ and
 its nice tree decomposition of width $\tw$ are given, and
 that $S=\{s_1,t_1,\ldots,s_k,t_k\}$ is the set of demand vertices.
It was shown in \cite{Scheffler94} that
 for each node $x$ in a nice tree decomposition $\mathcal{T}$,
 a solution $L$ of \kUDP,
 i.e., a set $L$ of disjoint paths connecting the demands
 is contained in $G_x$
 (the subgraph of $G$ induced by vertices and edges introduced in
 $x$ and its descendant nodes)
 in the following form.
\begin{lemma}[\cite{Scheffler94}, Lemmas 7 and~9]
\label{lm:UndirectedGxL}
Let $L$ be a solution of \kUDP.
For each node $x$,
 any connected component of the linear forest
 $G_x\cap L$ that is not a isolated vertex
 is an undirected $vw$-path of one of the following types.
\begin{enumerate}
\item
$v,w \in B_x\setminus S$.
\item
$v \in B_x\setminus S$ and $w\in S$.
\item
For some $i\in [k]$, $v=s_i$ and $w=t_i$.
\end{enumerate}
\end{lemma}
On the basis of Lemma~\ref{lm:UndirectedGxL},
 for each node $x$ and any possible partial solution $L_x$ in $G_x$,
 the algorithm of \cite{Scheffler94} computes
 information indicating that which type of connected component of $L_x$
 each vertex of the bag $B_x$ is contained in, by dynamic programming.
This information is represented by a mapping
 $\varphi:B_x\rightarrow
\{\mathsf 0,\mathsf 1\}\cup [k]\cup\{B_x\times\{\mathsf c,\mathsf d\}\}$.
For a vertex $v\in B_x$, $\varphi (v)$ indicates the following information.
\begin{description}
\item[$\varphi(v)=\mathsf 0$] indicates that
 $v$ is not contained in $L_x$ or an isolated vertex in $L_x$.
\item[$\varphi(v)=\mathsf 1$] indicates that
 $v$ is an internal vertex of a connected component of $L_x$
 or a non-isolated demand vertex in $S$, i.e.,
 $v$ is an internal vertex of a path of any type or an end-vertex in $S$
 of a path of type 2 or 3.
\item[{$\varphi(v)=i\in [k]$}] indicates that
 there is a connected component of $L_x$ connecting
 $v\in B_x\setminus S$ and $w\in\{s_i,t_i\}$
 but no connected component connecting
 $u\in B_x\setminus S$ and $w'\in\{s_i,t_i\}\setminus\{w\}$, i.e.,
 $v$ is an end-vertex of a path of type 2.
\item[$\varphi(u)=(v,\mathsf c)$ and $\varphi(v)=(u,\mathsf c)$] indicate that
 there is a connected component of $L_x$ connecting
 $u\in B_x\setminus S$ and $v\in B_x\setminus S$, i.e.,
 $u$ and $v$ are end-vertices of a path of type 1.
\item[$\varphi(u)=(v,\mathsf d)$ and $\varphi(v)=(u,\mathsf d)$] indicate that
 for some $i$,
 there are two connected components of $L_x$:
 one connecting $u\in B_x\setminus S$ and $w\in\{s_i,t_i\}\setminus B_x$
 and the other connecting
 $v\in B_x\setminus S$ and $w'\in\{s_i,t_i\}\setminus (B_x\cup\{w\})$,
 i.e.,
 $u$ and $v$ are internal vertices of a path of type 3.
\end{description}
Besides, we use the following fact obtained from the definition of
 the tree decomposition.
\begin{theorem}[\cite{Scheffler94}, Theorem~6]
\label{th:CutCondition}
For a node $x$,
 let $I_x$ be the set of demands $(s_i,t_i)$ such that
 $G_x$ has exactly one of $s_i$ and $t_i$.
If there is a node $x$ such that $|I_x|>|B_x|$, then
 \kUDP{} has no solution.
\end{theorem}
By Theorem~\ref{th:CutCondition},
 the number of $i$'s such that $\varphi(v)=i\in [k]$ is at most
 the maximum bag size $\tw+1$.
Therefore,
 the size of information retained at a node is at most
$(2+(\tw+1)+2\tw)^{\tw+1}=(3\tw+3)^{\tw+1}$.

This algorithm, as well as many other algorithms,
 first transforms a given tree decomposition to a nice tree decomposition
 $\mathcal T$ with $\mathcal O(n)$ nodes.
While the definition of a nice tree decomposition in \cite{Scheffler94}
 slightly differs from Definition~\ref{df:NiceTreeDecomposition},
 $\mathcal T$ can be computed in a time similar to that stated in 
 Proposition~\ref{pr:NiceTreeDecomposition}.

The bottleneck of the time complexity in the dynamic programming
 is the running time for a join node, which is at most
 $(3\tw+3)^{2(\tw+1)}\cdot\mathcal O(\tw)$.
Therefore, we can estimate the total running time as at most
$(3\tw + 3)^{2(\tw + 1)}\cdot\mathcal O(\tw\cdot n)
=\tw^{2\tw}\cdot 2^{\mathcal O(\tw)}\cdot\mathcal O(n)
=2^{2\tw\log\tw+\mathcal O(\tw)}\cdot n$.
If $\mathcal T$ is a path decomposition of width $\pw$, then
 because there is no join node, the running time can be reduced to
$(3\pw + 3)^{\pw + 1}\cdot\mathcal O(\pw\cdot n)
=2^{\pw\log\pw+\mathcal O(\pw)}\cdot n$.

The above discussion can be extended to \kDDP{} on a directed graph $G$.
\begin{lemma}
\label{lm:directedGxL}
Let $L$ be a solution of \kDDP{}.
For each node $x$,
 any connected component of the linear forest
 $G_x\cap L$ that is not a isolated vertex
 is a directed $vw$-path of one of the following types.
\begin{enumerate}
\item
$v,w \in B_x \setminus S$.
\item
$v \in B_x \setminus S$ and $w\in \{t_1,\ldots,t_k\}$.
\item
$w \in B_x \setminus S$ and $v\in \{s_1,\ldots,s_k\}$.
\item
For some $i\in [k]$, $v=s_i$ and $w=t_i$.
\end{enumerate}
\end{lemma}
As in the undirected case,
 for each node $x$,
 we can define a mapping
$\varphi:B_x\rightarrow\{\mathsf 0,\mathsf 1\}\cup[k]_{s}\cup[k]_{t}
 \cup
 \{B_x\times\{\mathsf c_{\textup{in}},\mathsf c_{\textup{out}},
 \mathsf d_{\textup{in}},\mathsf d_{\textup{out}}\}\}$
 representing information to be retained at $x$.
Here, $[k]_s$ and $[k]_t$ distinguish types 2 and 3 in
 Lemma~\ref{lm:directedGxL}, and
 $\mathsf c_{\textup{in}}$, $\mathsf c_{\textup{out}}$,
 $\mathsf d_{\textup{in}}$, and $\mathsf d_{\textup{out}}$
 distinguish the origin and destination of a path of types 1 and 4.

Because Theorem~\ref{th:CutCondition} also holds for \kDDP,
 the size of information retained at a node is at most
$(2+2(\tw+1)+4\tw)^{\tw+1}=(6\tw+4)^{\tw+1}$.
According to the modification of $\varphi$,
 the algorithm of \cite{Scheffler94} can naturally be modified to
 an algorithm for solving \kDDP.
As with the original algorithm,
 the time complexity can be estimated as
$(6\tw + 4)^{2(\tw + 1)}\cdot\mathcal O(\tw\cdot n)
=2^{2\tw\log\tw+\mathcal O(\tw)}\cdot n$, and
 if $\mathcal T$ is a path decomposition of width $\pw$, then
 it is reduced to $(6\pw + 4)^{\pw + 1}\cdot\mathcal O(\pw\cdot n)
=2^{\pw\log\pw+\mathcal O(\pw)}\cdot n$.
Thus, we obtain the following theorem.
\begin{theorem}
\kDDP{} can be solved in $2^{\pw\log\pw + \mathcal{O}(\pw)} \cdot n$ time.
\end{theorem}

%% file: section5.tex
\section{Conclusion}
\label{sc:Conclusion}
In this paper, we proposed a
 $2^{\mathcal O((\tw+k)\log k)}\cdot n$ time algorithm
 parameterized by the treewidth $\tw$ for solving the weighted version of \kDDP.
This running time is faster
 for a small $k=\tw^{o(1)}$
 than the previously known lower bound with $k=\Omega(\tw^4)$ under the ETH,
 and is
 a single exponential time of $\tw$ for a fixed $k$.
Since,
 as mentioned in Introduction,
 there is no
 subexponential time algorithm with respect to $\tw$ for any $k\geq 2$ under
 the ETH,
 the proposed algorithm achieves a tight running time for $k=\mathcal O(1)$.
Besides,
 the proposed algorithm
 can also solve \kUDP{} in the same running time with slight modifications.

As for the lower bound, we proved that, under the SETH,
 there is no
 $(2-\epsilon)^{\pw\log\pw}\cdot n^{\mathcal O(1)}$ time algorithm
 for solving \kDDP{} with a general $k$,
 and that
 a $2^{\pw\log\pw+\mathcal O(\pw)}\cdot n$ time algorithm can be obtained
 from a previously known algorithm with slight modifications, implying
 the proved lower bound is asymptotically tight.

An open problem is to close the remaining gap
 between the upper bound with $k=\omega(1)$ and the lower bound with
 $k=o(\tw^4)$ for \kDDP.